\documentclass[11pt]{article}
\usepackage{hyperref,graphicx}
\usepackage{fullpage}

\usepackage{tikz}
\usepackage{complexity}

\usetikzlibrary{arrows}
\usepackage{amsmath,amsfonts,amsthm}
\usepackage{float}
\usepackage{wrapfig}
\usepackage{caption,subcaption}

\newcommand{\poset}[1]{\ensuremath{{#1}}}
\newcommand{\lattice}[1]{\ensuremath{{#1}}}
\newcommand{\circuit}[1]{\ensuremath{#1}}

\newcommand{\family}[1]{\ensuremath{{#1}}}
\newcommand{\setunion}{\ensuremath{\cup}}
\newcommand{\sd}[2]{\class{SIZE\lower-0.12em\hbox{-}DEPTH}({#1}, {#2})}
\newcommand{\pleq}[1]{\ensuremath{\leq_\poset{P}}}
\newcommand{\var}[1]{\ensuremath{\mathit{#1}}}

\newcommand{\skewcc}{\class{SkewCC}}
\newcommand{\skewccvp}{\lang{SkewCCVP}}
\newcommand{\dgapp}{\lang{DGAP1'}}
\newcommand{\dgap}{\lang{DGAP1}}
\renewcommand{\L}{\class{LOG}}
\renewcommand{\NL}{\class{NLOG}}

\newcommand{\PiSkewCCVP}{\ensuremath{\lang{\Pi_i\lower-0.12em\hbox{--}SkewCCVP}}}

\newcommand{\CCVP}[2]
           {\ensuremath{(\poset{#1}, #2)\lower-0.12em\hbox{--}\lang{CCVP}}}
\newcommand{\PCC}[2]
           {\ensuremath{(\poset{#1}, #2)\lower-0.12em\hbox{--}\class{CC}}}
\newcommand{\PskewCC}[2]
           {\ensuremath{(\poset{#1}, #2)\lower-0.12em\hbox{--}\class{SkewCC}}}
\newcommand{\formula}[2]
           {\ensuremath{(\poset{#1}, #2)\lower-0.12em\hbox{--}\class{Formulae}}}
\newcommand{\XCC}[1]
           {\ensuremath{\poset{#1}\lower-0.12em\hbox{--}\class{CC}}}
\newcommand{\XSkewCC}[1]
           {\ensuremath{\poset{#1}\lower-0.12em\hbox{--}\class{SkewCC}}}
\newcommand{\FCC}[1]
           {\ensuremath{\class{\family{#1}\lower-0.12em\hbox{--}CC}}}

\newcommand{\smalllabel}[1]{{\tiny #1}}

\newtheorem{definition}{Definition}
\newtheorem{theorem}{Theorem}
\newtheorem{lemma}{Lemma}
\newtheorem{corollary}{Corollary}
\newtheorem{claim}{Claim}
\newtheorem{proposition}{Proposition}
\newtheorem{remark}{Remark}

\newcommand{\dist}[2]{{\tt DIST}_{#1, #2}}
\newcommand{\checkge}[1]{{\tt GE}_{#1}}
\newcommand{\checkgeb}[1]{{\tt GE'}_{#1}}

\newlength{\arrsize}  
\pgfarrowsdeclare{biggertip}{biggertip}{  
  \setlength{\arrsize}{0.8pt}  
  \addtolength{\arrsize}{.5\pgflinewidth}  
  \pgfarrowsrightextend{0}  
  \pgfarrowsleftextend{-5\arrsize}  
}{  
  \setlength{\arrsize}{0.8pt}  
  \addtolength{\arrsize}{.5\pgflinewidth}  
  \pgfpathmoveto{\pgfpoint{-5\arrsize}{4\arrsize}}  
  \pgfpathlineto{\pgfpointorigin}  
  \pgfpathlineto{\pgfpoint{-5\arrsize}{-4\arrsize}}  
  \pgfusepathqstroke  
}

\title{Comparator Circuits over Finite Bounded Posets}
\author{Balagopal Komarath\thanks{Supported by TCS PhD Fellowship} \hspace{1cm} Jayalal Sarma \hspace{1cm} K.S. Sunil \\[3mm]
{\large Department of Computer Science \& Engineering.} \\ 
{\large Indian Institute of Technology Madras, Chennai, India.} \\[1mm]
{\large Email: \{{\tt {baluks|jayalal|sunil}\}@cse.iitm.ac.in}}}

\begin{document}
\maketitle
\vspace{-5mm}

\begin{abstract}

  The comparator circuit model was originally introduced in \cite{MayrS92}
  (and further studied in \cite{Cook12}) to capture problems that are
  not known to be $\P$-complete but still not known to admit efficient
  parallel algorithms. The class $\CC$ is the complexity class of
  problems many-one logspace reducible to the Comparator Circuit Value
  Problem and we know that $\NL \subseteq \CC \subseteq \P$. Cook {\em
    et al.} \cite{Cook12} showed that $\CC$ is also the class of
  languages decided by polynomial size comparator circuit families.

  We study generalizations of the comparator circuit model that work
  over fixed finite bounded posets. We observe that there are
  universal comparator circuits even over arbitrary fixed finite
  bounded posets. Building on this, we show the following :
  \begin{itemize}
  \item Comparator circuits of polynomial size over fixed finite
    \textit{distributive} lattices characterize the class $\CC$. When
    the circuit is restricted to be skew, they characterize
    $\L$. Noting that (uniform) polynomial sized Boolean circuits
    (resp. skew) characterize $\P$ (resp. $\NL$), this indicates a
    comparison between $\P$ vs $\CC$ and $\NL$ vs $\L$ problems.
  \item Complementing this, we show that comparator circuits of
    polynomial size over arbitrary fixed finite lattices characterize
    the class $\P$ even when the comparator circuit is skew.
  \item In addition, we show a characterization of the class $\NP$ by
    a family of polynomial sized comparator circuits over fixed {\em
      finite bounded posets}. As an aside, we consider generalizations
    of Boolean formulae over arbitrary lattices. We show that Spira's
    theorem~\cite{Spira71} can be extended to this setting as well and
    show that polynomial sized Boolean formulae over finite fixed
    lattices capture the class $\NC^1$.
\end{itemize}
These results generalize results in \cite{Cook12} regarding the power
of comparator circuits. Our techniques involve design of comparator
circuits and finite posets. We then use known results from lattice
theory to show that the posets that we obtain can be embedded into
appropriate lattices. Our results give new methods to establish $\CC$
upper bounds for problems and also indicate potential new approaches
towards the problems $\P$ vs $\CC$ and $\NL$ vs $\L$ using lattice
theoretic methods.

\end{abstract}

\newpage

\section{Introduction}


Completeness for the class $\P$ for a problem is usually considered
to be evidence that it is hard to design an efficient parallel
algorithm for the problem. However, there are many computational
problems in the class $\P$, which are not known to be $\P$-complete,
yet designing efficient parallel algorithms for them has remained
elusive. Some of the classical examples of such problems include
lex-least maximal matching problem and stable marriage problem
\cite{MayrS92}.


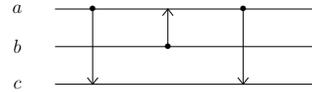
\begin{wrapfigure}{r}{0.4\textwidth}
\begin{center}
\scalebox{0.5}{%
\begin{tikzpicture}
\node (c1) at (1, 3) {\textbullet};
\node (c2) at (1, 1) {};
\draw [-biggertip] (c1.center) -- (c2.center);
\node (c3) at (3, 3) {};
\node (c4) at (3, 2) {\textbullet};
\draw [-biggertip] (c4.center) -- (c3.center);
\node (c7) at (5, 3) {\textbullet};
\node (c8) at (5, 1) {};
\draw [-biggertip] (c7.center) -- (c8.center);

\node (a) at (-1, 3) {\Large$a$};
\draw (0,3) -- ++(7, 0);
\node (b) at (-1, 2) {\Large$b$};
\draw (0,2) -- ++(7, 0);
\node (x) at (-1, 1) {\Large$c$};
\draw (0,1) -- ++(7, 0);
\end{tikzpicture}
}
\caption{A Comparator Circuit}
\label{fig:comparator}
\end{center}
\end{wrapfigure}

Attempting to capture the exact complexity of computation in these
problems using a variant of Boolean circuit model, Mayr and
Subramanian \cite{MayrS92} (see also \cite{Cook12}) studied the
comparator circuit model.  A comparator circuit is a sorting network
working over the values 0 and 1. A comparator gate has two inputs and
two outputs. The first output is the AND of the two inputs and the
second output is the OR of the two inputs. A comparator circuit is a
circuit that has only comparator gates. In particular, fan-out gates
are not allowed. Without loss of generality, we can assume that NOT gates are used only at
the input level. A graphical representation of a comparator circuit is
shown in Figure~\ref{fig:comparator}. In this representation, we draw
a set of parallel \emph{lines}. Each line carries a logical value
which is updated by gates incident on that line. Each gate is
represented by a directed arrow from one line (say $i$) to another
(say $j$) and the gate updates the values of lines as follows. The
value of line $i$ ($j$) is set to the AND (resp. OR) of values
previously on lines $i$ and $j$. The gates are evaluated from left to
right. The output of the circuit is the final value of a line
designated as the output line. We define the model formally in Section~\ref{sec:prelims}.

In order to study the complexity theoretic significance of comparator
circuits, the corresponding circuit value problem was explored in
\cite{MayrS92}. That is, given a comparator circuit and an input, test
if the output wire carries a 1 or not. The class \CC\ is defined in
\cite{MayrS92} as the class of languages that are logspace many-one
reducible to the comparator circuit value problem.  They also observed
that the class \CC\ is contained in \P.  Feder's algorithm (described
in \cite{Subr90}) for directed reachability proves that the class \CC\
contains \NL\ as a subclass. These are the best containments currently
known about the complexity class \CC.


There has been a recent spurt of activity in the characterization of
$\CC$. Cook et al.~\cite{Cook12} showed that the class $\CC$ is robust
even if the complexity of the many-one reduction to the comparator
circuit value problem is varied from $\AC^0$ to $\NL$.  They also gave
a characterization of the class \CC\ in terms of a computational model
(comparator circuit families). Their main contribution in this regard
is the introduction of a universal comparator circuit that can
simulate the computation of a comparator circuit given as input (to
the universal circuit).  Comparison of $\CC$ with the class $\NC$ has
interesting implications to the corresponding computational
restrictions. For example, hardness for the class $\CC$ is conjectured
to be evidence that the problem is not efficiently
parallelizable. This intuition was further strengthened by Cook et
al.~\cite{Cook12} by showing that there are oracle sets relative to
which $\CC$ and $\NC$ are incomparable (\NC\ is the class of all
languages efficiently solvable by parallel algorithms). In addition,
it is conjectured in \cite{Cook12} that the classes $\NC$, $\SC$ and
$\CC$ are pairwise incomparable.

\paragraph{\textbf{Our Results \& Techniques:}}
In this paper, we study the computational power of comparator circuits
working over arbitrary fixed finite bounded posets. Informally,
instead of 0 and 1, the values used throughout the computation could be any
element from the poset and the AND and OR gates compute maximal lower
bounds and minimal upper bounds over the poset respectively. We define
this model formally in section~\ref{sec:general}.  We obtain the
following results:

\begin{itemize}
\item There exist Universal Comparator Circuits for comparator
  circuits irrespective of the underlying bounded
  poset. (Proposition~\ref{prop:univ-cir}, Section~\ref{sec:general}.)

\item Comparator circuits of polynomial size over fixed finite
  distributive lattices capture the class \CC. (Theorem~\ref{thm:CC},
  Section~\ref{sec:P}). This leads to a new way to show that a problem
  is in the class $\CC$. That is, by designing a comparator circuit
  over a fixed finite lattice and then showing that the lattice is
  distributive. (An application of this method to design $\CC$
  algorithms for the stable matching problem can be found in
  \cite{MayrS92}. See also Section {6.2} in \cite{Cook12}). Since
  there are lattice theoretic techniques known (cf. $M_3$-$N_5$
  Theorem \cite{lattice-text}) for showing that a lattice is
  distributive, this alternate definition of the class $\CC$ using
  comparator circuits over distributive lattices might be of
  independent interest.

\item Going beyond distributivity, we show that comparator circuits of
  polynomial size over fixed finite lattices characterize the class
  \P. (Theorem~\ref{thm:P}, Section~\ref{sec:P}). In particular, we
  design a fixed finite poset $P$ over which, for any language $L \in
  \P$, there is a polynomial size comparator circuit family over $P$
  computing $L$. During computation, we only use lubs and glbs that
  exist in the poset $P$. This enables the use of Dedekind-MacNeille
  completion (DM completion) to construct a fixed finite
  lattice completing the poset $P$ while preserving existing lubs and
  glbs in the poset and that lattice can be used to perform all
  computations in \P. A potential drawback of the lattice thus
  obtained is that the complexity class captured by comparator
  circuits over it may vary depending on the element in the lattice
  used as the accepting element. By using standard tools from lattice
  theory, we derive that there is a fixed constant $i \ge 3$, such
  that comparator circuits over $\Pi_i$ (where $\Pi_i$ is the $i^{th}$
  partition lattice - see Section~\ref{sec:prelims} for a definition)
  with polynomial size can compute all functions in $\P$. Moreover, we
  show that comparator circuits over the lattice $\Pi_{i}$ capture
  $\P$ irrespective of the accepting element used.
  
  However, both partition lattices for $i \ge 3$ and the lattice given
  by DM completion are non-distributive. Exploring the possibility of
  another completion of the poset $P$ into a distributive lattice that
  preserves existing lubs and glbs (which will show $\P = \CC$), we
  arrive at the following negative result : the poset $P$ cannot be
  embedded into any distributive lattice while preserving all
  existing lubs and glbs. (Theorem~\ref{thm:impossibility}).
  
  It is conceivable that the class \P\ could be captured by a family
  of distributive lattices, while no finite fixed lattice capturing
  \P\ can be distributive. Motivated by this, we also present an
  analogue of the main theorem using growing posets of much simpler
  structure (See appendix~\ref{app:P-Growing}). However, we argue that
  this poset family also cannot be embedded into a family of
  distributive lattices while preserving all existing lubs and glbs.

\item Going beyond lattice structure, we show that comparator circuits
  over fixed finite bounded \textit{posets} capture the class
  \NP. (Theorem~\ref{thm:NP}, Section~\ref{sec:NP}). Here, we
  crucially use the fact that posets that are not lattices could have
  elements that do not have unique minimal upper bounds to simulate
  non-determinism. Hence, any completion of this poset into a lattice
  will fail to capture \NP, unless $\P = \NP$. Note, that the DM
  completion of this poset would not be able to characterize \NP\ as
  in the case of \P\ because the DM completion would introduce
  elements so that the elements in the poset that have non-unique
  minimal upper bounds and/or maximal lower bounds would end up having
  unique lubs and glbs.

\item Restricting the structure of the comparator circuit, we obtain
  an exact characterization of the class $\L$ using skew comparator
  circuits (Theorem~\ref{thm:L}). Noting that the polynomial sized
  skew Boolean circuits characterize exactly the class $\NL$, this
  leads to a comparison between $\CC$ vs $\P$ and $\NL$ vs $\L$
  problem : \textit{both problems address the power of polynomial size
    Boolean circuits vs comparator circuits in general and skew
    circuits respectively}.
\item We further study generalizations of skew comparator circuits to
  arbitrary lattices. When the lattice is distributive, it follows
  that the circuits capture exactly $\L$. Complementing this, we show
  that there are fixed finite lattices $P$ over which the skew
  comparator circuits characterize exactly
  $\P$.(Theorem~\ref{thm:SkewP}).
\item We study generalizations of Boolean formulas to arbitrary
  lattices where the AND and OR gates compute the $\land$ and $\lor$
  of the lattices. We generalize Spira's theorem~\cite{Spira71} to this
  setting and show that polynomial sized Boolean formulae over finite
  fixed lattices capture exactly $\NC^1$ (Theorem~\ref{thm:NC1}).
\end{itemize}

Thus, we observe that as the comparator circuit is allowed to compute
over progressively general structures (from distributive lattices to
arbitrary lattices to posets), the model captures classes of problems
that are progressively harder to parallelize (From \CC\ to \P\ to
$\NP$). The table below indicates the results (known results are
indicated by citations).
\begin{center}
\begin{tabular}{|l|c|c|c|c|}
\hline
\textbf{Lattices $\implies$} & \textbf{Boolean} & \textbf{Distributive} & \textbf{General} & \textbf{Posets} \\
\hline
$\poly$-sized & $\CC$(see \cite{Cook12}) & $\CC$ & $\P$ & $\NP$ \\
\hline
Skew, $\poly$-sized  & $\L$ & $\L$ & $\P$ & - \\
\hline
Formulae & $\NC^1$(see \cite{Spira71}) & $\NC^1$ & $\NC^1$ & \\
\hline
\end{tabular}
\end{center}

The main technical contribution in our proofs is the design of posets
and the corresponding comparator circuits for capturing complexity
classes. We then use known ideas from lattice and order theory in
order to derive lattices to which the constructed posets can
be embedded.

\section{Preliminaries}
\label{sec:prelims}
The standard definitions in complexity theory used in this paper can
be found in standard textbooks \cite{arora-barak}. All reductions in
this paper are computable in logspace unless mentioned otherwise. By
(standard) Boolean circuits, we mean circuits over the basis $\{ \vee,
\wedge \}$ where NOT gates are only allowed at the input level. In
this section, we define comparator circuits, certain restrictions on
comparator circuits and complexity classes based on those
restrictions.

A \emph{comparator circuit} has a set of $n$ \emph{lines} $\{ w_1,
\ldots ,w_n \}$ and an ordered list of gates $(w_i, w_j)$. Each line
can be fed as input a value that is either (Boolean) 0 or 1. We define
$val(w_i)$ to be the value of the line $w_i$. Each gate $(w_i, w_j)$
updates $val(w_i)$ to $val(w_i) \wedge val(w_j)$ and $val(w_j)$ to
$val(w_i) \vee val(w_j)$ in order. After all gates have updated the
values, the value of the line $w_1$ is the output of the circuit.

The {\sc Comparator Circuit Value problem} is: Given $(C, x)$ as input
find the output of the comparator circuit \circuit{C} when fed $x$ as
input.  We can think of $C$ being encoded according to the above
definition of comparator circuits. We call this the ordered list
representation as the gates are presented as an ordered list. Mayr and
Subramanian~\cite{MayrS92} defined the complexity class \CC\ as the
set of all languages logspace reducible to the Comparator Circuit
Value problem. Cook et al.~\cite{Cook12} characterized the class
\CC\ as languages computed by $\AC^0$-uniform families of annotated
comparator circuits. In an annotated comparator circuit the initial
value of a line could be an input variable $x_i$ or its complement
$\overline{x_i}$. In a family of annotated comparator circuits for a
language \lang{L}, the $n^{th}$ comparator circuit in the family has
exactly $n$ input variables ($x_1, \ldots ,x_n$) and the circuit
computes $\lang{L} \cap {\{0, 1\}}^{n}$.

\paragraph{Skew Comparator Circuits:} 
We now define skewness in comparator circuits. To begin with, we
present an alternate definition of comparator circuits that is closer
to the definition of standard Boolean circuits. A comparator gate is a
2-input, 2-output gate that takes $a$ and $b$ as inputs and outputs $a
\wedge b$ and $a \vee b$. Then the comparator circuit is simply a
circuit (in the usual sense) that consists of only comparator gates
(In particular, fan-out gates are not allowed). Using this definition,
we can encode comparator circuits by using DAGs as we encode standard
Boolean circuits. It is easy to see that given a comparator circuit
encoded as an ordered list of gates, we can obtain the DAG encoding
the comparator circuit in logspace. Using this definition, we can talk
about \emph{wires} in the comparator circuit.

We say that an AND gate in a comparator gate is \textit{used} if the
AND output wire of that comparator gate has a path, through comparator
gates, to the output wire. An AND gate in the circuit is called
\emph{skew} if and only if at least one input to that gate is the
constant 0 or the constant 1 or (in the case of annotated circuits) an
input bit $x_i$ or $\overline{x_i}$ for some $i$.

A comparator circuit is called a \emph{skew comparator circuit} if and
only if all used AND gates in the circuit are \emph{skew}.
%
The complexity class \skewcc\ consists of all languages that can be
decided by poly-size skew comparator circuit families.
%
We define \skewccvp\ to be the circuit evaluation problem for skew
comparator circuits. Note that given the ordered list representation
of a comparator circuit, it is easy to check whether an AND gate is
used or not. For ex., if the $i^{\hbox{th}}$ gate is $(w_{1}, w_{2})$,
then the AND output of this gate is unused if and only if there is no
element in the list of gates with $w_{1}$ as a member at a position
greater that $i$ in the list and if the AND output of this gate is not
the output wire.

The circuit family is $\L$-uniform if and only if there exists a TM
$\lang{M}$ that outputs the $n^\textrm{th}$ circuit in the family in
$O(\log(n))$ space given $1^{n}$ as input. All circuits in this paper
are $\L$-uniform unless mentioned otherwise.


\paragraph{Lattice and Order Theory:}
We include some basic definitions and terminology from standard
lattice and order theory that are required later in the paper. A more
detailed treatment can be found in standard textbooks
\cite{lattice-text}. 

A set \poset{P} along with a reflexive, anti-symmetric and transitive
relation denoted by $\leq_{\poset{P}}$ is called a {\em poset}. An
element $m \in \poset{P}$ is called the {\em greatest element} if $x
\leq m$ for all $x$ in \poset{P}.  An element $m \in \poset{P}$ is
called the {\em least element} if $m \leq x$ for all $x$ in \poset{P}.
A poset is called {\em bounded} if it has a greatest and a least
element. Note that any finite poset can be converted into a finite
bounded poset by adding two new elements 0 and 1 and adding the
relations $m \leq 1$ and $0 \leq m$ for every element $m$ in the
poset.  {\em Minimal upper bounds} of two elements $x$, $y$ in
\poset{P}, denoted by $x \vee y$, is the set of all $m \in \poset{P}$
such that $x \leq m$, $y \leq m$ and there exists no $m'$ distinct
from $m$ such that $x \leq m'$, $y \leq m'$ and $m' \leq m$.  {\em
  Maximal lower bounds} of two elements $x$, $y$ in \poset{P}, denoted
by $x \wedge y$, is the set of all $m \in \poset{P}$ such that $m \leq
x$, $m \leq y$ and there exists no $m'$ distinct from $m$ such that
$m' \leq x$, $m' \leq y$ and $m \leq m'$.  A poset \poset{P} is called
a {\em lattice} if every pair of elements $x$ and $y$ has a unique
maximal lower bound and a unique minimal upper bound. In a lattice,
the minimal upper bound (maximal lower bound) of two elements is also
known as the {\em join} ({\em meet}). Since minimal upper bound and
maximal lower bound are unique in a lattice, we drop the set notation
when describing them, i.e., instead of writing $a \vee b = \{ x \}$,
we simply write $a \vee b = x$.  A lattice \lattice{L} is called {\em
  distributive} if for every elements $a$, $b$, $c$ $\in \lattice{L}$
we have $a \vee (b \wedge c) = (a \vee b) \wedge (a \vee c)$.  An
\emph{order embedding} of a poset \poset{P} into a poset \poset{P'} is
a function $f : \poset{P} \mapsto \poset{P'}$ such that $f(x)
\leq_{\poset{P'}} f(y) \iff x \leq_{\poset{P}} y$. We say that the
lattice $L$ is a sub-lattice of $L'$ if $L \subseteq L'$ and $L$ is
also a lattice under the meet and join operations inherited from
$L'$. In this case, we say that \emph{$L'$ embeds $L$}.

We now state some technical theorems from the theory which we
crucially use. The following theorem shows that given a poset one can
find a lattice that contains the poset.

\begin{theorem}[Dedekind-Macneille Completion\cite{lattice-text}]
For any poset \poset{P}, there always exists a smallest lattice
\lattice{L} that order embeds \poset{P}. This lattice \lattice{L} is
called the {\em Dedekind-MacNeille completion} of \poset{P}.
\label{thm:dmcompletion}
\end{theorem}

One crucial property of Dedekind-MacNeille completion is that it
preserves all meets and joins that exist in the poset, i.e., if $a$
and $b$ are two elements in the poset and $a \vee b = x$ in the poset,
then we have $f(a) \vee f(b) = f(x)$ in the Dedekind-MacNeille
completion of the poset, where $f$ is the embedding function that maps
elements in $P$ to elements in $L$.

We now state a very important theorem that concerns the structure of
distributive lattices.

\begin{theorem}[Birkhoff's Representation Theorem\cite{lattice-text}]
The elements of any finite distributive lattice can be represented as
finite sets, in such a way that the join and meet operations over the
finite distributive lattice correspond to unions and intersections of
the finite sets used to represent those elements.
\label{thm:birkhoff}
\end{theorem}

The $n^{\hbox{th}}$ \emph{partition lattice} for $n \geq 2$, denoted
$\Pi_{n}$, is the lattice where elements are partitions of the set
$\{1, \ldots ,n\}$ ordered by refinement. Equivalently, the elements
are equivalence relations on the set $\{1, \ldots ,n\}$ where the glb
is the intersection and lub is the transitive closure of the union.

\begin{theorem}[Pudl{\'a}k, T{\r{u}}ma\cite{pudlak80}]
  For any finite lattice $L$, there exists an $i$ such that $L$ can be
  embedded as a sublattice in $\Pi_{i}$.
\label{thm:pudlak}
\end{theorem}

We can describe elements of the partition lattice $\Pi_{n}$ by using
undirected graphs on the vertex set $\{1, \ldots ,n\}$. Given an
undirected graph $G = (\{1, \ldots ,n\}, E)$, the corresponding
element $A_{G} \in \Pi_{n}$ is the equivalence relation $A_{G} = \{(i,
j) : j \text{ is reachable from } i \text{ in } G \}$. We may choose 
transitively closed graphs (disjoint union of cliques) as the
canonical representation for elements of partition lattices.
Figure~\ref{fig:part_lattice} shows the lattice $\Pi_4$ and two
undirected graphs representing two different elements in $\Pi_4$.

\begin{figure}
\centering
\includegraphics[scale=0.5]{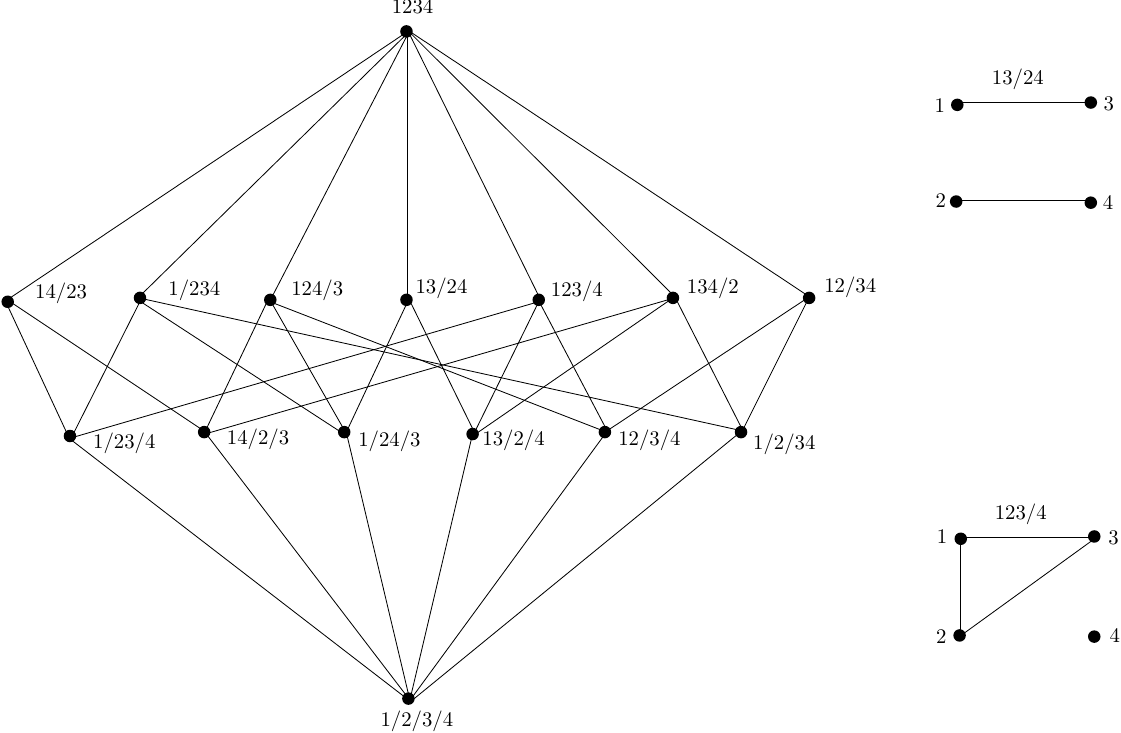}
\caption{The lattice $\Pi_4$ and the undirected graph representation of $13/24$ and $123/4$}
\label{fig:part_lattice}
\end{figure}

\paragraph{Some Relations in Partition Lattices:}
A formula over a lattice is defined analogously to a Boolean
formula. The Boolean AND and OR operations are generalized to glb and
lub operations of the lattice and the formula may contain elements of
the lattice as constants (Similar to Boolean values 0 and 1 in a
Boolean formula). In this section, we prove the existence of a certain
formula over partition lattices. The following statements
hold\footnote{Since we have not seen them explicitly in the
  literature, we include the proofs in this paper.}  for any partition
lattice $\Pi_{i}$ where $i \geq 2$. In the following propositions, the
element 0 refers to the least element of the lattice and the element 1
refers to the greatest element of the lattice.

\begin{proposition}
  For any $A, B \in \Pi_{i}$ such that $A \not\leq B$,
  there exists a formula $\dist{A}{B}$ over $\Pi_{i}$ such that
  $\dist{A}{B}(x) = 1$ if $x = A$ and strictly less than $1$ if $x =
  B$.
  \label{prop:dist}
\end{proposition}
\begin{proof}
  There are two cases to consider.
  Case when [$A > B$]: Let $P \in \Pi_{i}$ be the element corresponding to
    a path with exactly one vertex from each partition in $A$. We
    define $\dist{A}{B}(x) = x \vee P$.
    Case when [$A \not> B$]: Let $e_{1}, \ldots ,e_{m}$ be the edges in
    $B \backslash A$ (using canonical representation) and let $B_{i}$
    denote the element in the partition lattice that correspond to the
    undirected graph having only the edge $e_{i}$. Let $g(x) = x \vee
    B_{1} \vee \ldots \vee B_{m}$. We have $g(A) > g(B) = B$. Then
    define $\dist{A}{B}(x) = \dist{g(A)}{B}(g(x))$.
\end{proof}

\begin{proposition}
  For any $A \in \Pi_{i}$, there exists a formula $\checkge{A}(x)$
  that is 1 iff $x \geq A$. In addition, there exists a formula
  $\checkgeb{A}(x)$ that evaluates to 1 if $x \geq A$ and evaluates to
  0 otherwise.
  \label{prop:checkgeb}
\end{proposition}
\begin{proof}
  For the first part, simply define the formula $\checkge{A}(x) =
  \bigwedge_{B \not\geq A} \dist{A}{B}(x)$ when $A \neq 0$ where
  $\dist{A}{B}(x)$ is as defined in
  Proposition~\ref{prop:dist}. Define $\checkge{0}$ as identically 1.

  For the second part, consider the formula $f_{Z}$ that is defined if
  and only if $Z \neq 1$ and it maps 1 to 1 and $Z$ to 0 (the images
  of the rest of the elements in the lattice can be arbitrary). Let
  $Z$ have $k \geq 2$ partitions. Let $e_{1}, \ldots ,e_{m}$ be the
  edges of the complete $k$-partite graph on these $k$ partitions. Let
  $B_{1}, \ldots ,B_{m}$ be lattice elements such that $B_{i}$
  corresponds to the undirected graph that contains only the edge
  $e_{i}$.
 $ f_{Z}(x) = \bigvee_{i = 1}^{m} \left(x \wedge B_{i}\right)$.
 Now to complete the second part, define the formula $\checkgeb{A}(x) = \bigwedge_{B < 1}
  f_{B}(\checkge{A}(x))$ ($\checkgeb{0}$ is identically 1).
\end{proof}

\section{Generalization to Finite Bounded Posets and Universal Circuits}
\label{sec:general}

In this section, we consider comparator circuit models over arbitrary
fixed finite bounded posets instead of the Boolean lattice on two
elements. We then prove the existence of universal circuits for these
models. The existence of these generalized universal comparator
circuits imply that the classes characterized by comparator circuit
families over fixed finite bounded posets also have canonical complete
problems -- the comparator circuit evaluation problem over the same
fixed finite bounded poset.

\begin{definition}[Comparator Circuits over Fixed Finite Bounded Posets]
  A comparator circuit family over a finite bounded poset \poset{P}
  with an accepting element\footnote{In the definition of general
    Boolean circuits it is implicit that the element 1 is the
    accepting element. However, it does not make any difference even
    if we use 0 as the accepting element. This is because a Boolean
    circuit that accepts using 0 can be easily converted to one that
    accepts using 1 by complementing the output. This is not true for
    comparator circuits over bounded posets in general. Using
    different elements as accepting elements may change the power of
    the comparator circuit.} $a \in \poset{P}$ is a family of circuits
  $\family{C} = {\{ \circuit{C_n} \}}_{n \geq 0}$ where
  $\circuit{C_n}$ $=$ $(W, G, f)$ and $f : W \mapsto (\poset{P}
  \setunion \{ (i, g) : 1 \leq i \leq n \textrm{ and } g:\Sigma
  \mapsto \poset{P} \})$. Here $W = \{ w_1,
  \ldots , w_m \}$ is a set of lines and $G$ is an ordered list of
  gates $(w_i, w_j)$.

  On input $x \in \Sigma^n$, we define the output of the comparator
  circuit $\circuit{C_n}$ as follows.  Each line is initially assigned
  a value according to $f$ as follows. We denote the value of the line
  $w_i$ by $val(w_i)$.  If $f(w) \in \poset{P}$, then the value is the
  element $f(w)$.  Otherwise $f(w) = (i, g)$ and the initial value is
  given by $g(x_i)$.  A gate $(w_i, w_j)$ (non-deterministically)
  updates the value of the line $w_i$ into $val(w_i) \wedge val(w_j)$
  and the value of the line $w_j$ into $val(w_i) \vee val(w_j)$. The
  values of lines are updated by each gate in $G$ in order and the
  circuit accepts $x$ if and only if $val(w_1) = a$ at the end of the
  computation for some sequence of non-deterministic choices.

  Let $\Sigma$ be any finite alphabet. A comparator circuit family
  \family{C} over a bounded poset \poset{P} with an accepting element
  $a \in \poset{P}$ decides $\lang{L} \subseteq \Sigma^*$ if
  $\circuit{C_{|x|}}(x)=a$ if and only if $x \in \lang{L}$ $\forall x
  \in \Sigma^*$.

  All comparator circuit families in this paper are logspace-uniform
  unless mentioned otherwise.
\label{def:comp}
\end{definition}

\begin{remark}
  Note that when the underlying poset is a lattice, the output of all
  gates in the comparator circuit is deterministic. In other words,
  the non-determinism in the circuit comes from the fact that two
  elements in a poset need not have unique lubs and glbs.
\end{remark}

Note that we can generalize any circuit model that uses only AND and
OR gates to work over arbitrary bounded posets. However, as we will
see in this paper, the most interesting case is comparator circuits
over arbitrary bounded posets as they lead to new characterizations of
complexity classes other than \CC.

\begin{definition}
  We define the complexity class $\PCC{P}{a}$ as the set of all
  languages accepted by poly-size comparator circuit families over the
  finite bounded poset \poset{P} with accepting element $a \in
  \poset{P}$.
\end{definition}

If the complexity class does not change with the accepting element,
i.e., $\PCC{P}{a} = \PCC{P}{b}$ for any $a, b \in P$, we simply write
$\XCC{P}$ to refer to the complexity class $\PCC{P}{a}$.

We note that for any bounded poset \poset{P} with at least 2 elements,
we can simulate a Boolean lattice by using 0 (least element) and some
$a > 0$ in \poset{P}. Therefore, we have $\CC \subseteq \PCC{P}{a}$.

\begin{definition}
  For any finite bounded poset \poset{P} and any $a \in \poset{P}$,
  the comparator circuit evaluation problem \CCVP{P}{a} is defined as
  the set of all tuples $(C, x)$ such that $C$ on input $x$ has a
  sequence of non-deterministic choices where it outputs $a \in
  \poset{P}$ where \circuit{C} is a comparator circuit over \poset{P}.
\end{definition}

We now describe an encoding for the $\CCVP{P}{a}$ problem. The input
is encoded by a binary string of the form $1^{n}01^{m}0\{0,
1\}^{n(n-1)m + n}$. Here the last $n(n-1)m$ bits of the string can be
viewed as $m$ blocks of $n(n-1)$ bits where the $i^{\hbox{th}}$ block
has exactly one set bit, say $(k, j)$ where $k \neq j$, and it encodes
the fact that the $i^{\hbox{th}}$ gate is from line $k$ to line
$j$. The $n$ bits prior to these bits encode the initial values of $n$
lines. This encoding is logspace-equivalent to the ordered list
representation. We call strings of this form $(n, m)$-valid.  A given
$N$-bit string can be valid for at most one $(n, m)$ pair.  We first
prove that a universal comparator circuit exists even for comparator
circuit model working over arbitrary finite fixed posets.

\begin{proposition}
  \label{prop:univ-cir}
  For any bounded poset \poset{P}, there exists a universal comparator
  circuit $\circuit{U_{n,m}}$ over \poset{P} that when given
  $(\circuit{C}, x)$ as input, where $\circuit{C}$ is a comparator
  circuit over \poset{P} with $n$ lines and $m$ gates, simulates the
  computation of \circuit{C}. That is, \circuit{U_{n,m}} has a
  non-deterministic path that outputs $a \in \poset{P}$ if and only if \circuit{C}
  has such a path, for any $a \in \poset{P}$. Moreover, the size of
  \circuit{U_{n,m}} is $\poly(n,m)$.
\end{proposition}
\begin{proof}
  We simply observe that the construction for a universal circuit for
  the class \CC\ in \cite{Cook12} generalizes to arbitrary bounded
  posets. The gadget shown in Figure~\ref{fig:cond} simulates the
  comparator gate $g = (y, x)$ depending on the ``enable'' input
  $e$. Here the inputs $e$ and $\overline{e}$ satisfy the following
  property. If $e = 0$ ($e = 1$), then $\overline{e} = 1$
  ($\overline{e} = 0$ resp) where $0$ and $1$ are the least and
  greatest elements of the bounded poset \poset{P} respectively.  If
  the enable input is 1, then gate $g$ is active. If the enable input
  is 0, then the gate $g$ acts as a pass-through gate, i.e., the lines
  labelled $x$ and $y$ retain their original values.
  
  Now to simulate a single gate in the circuit \circuit{C}, the
  universal circuit uses $n(n-1)$ such gadgets where $n$ is the number
  of lines in $\circuit{C}$. The inputs $e$ and $\overline{e}$ for
  each gadget is set according to \circuit{C}. The circuit \circuit{C}
  can be simulated using $n(n-1)m$ gates where $m$ is the number of
  gates in \circuit{C}.
\end{proof}
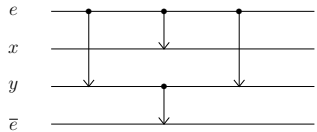
\begin{wrapfigure}{r}{0.5\textwidth}
\begin{center}
\scalebox{0.5}{%
\begin{tikzpicture}
\node (c1) at (1, 3) {\textbullet};
\node (c2) at (1, 1) {};
\draw [-biggertip] (c1.center) -- (c2.center);
\node (c3) at (3, 3) {\textbullet};
\node (c4) at (3, 2) {};
\draw [-biggertip] (c3.center) -- (c4.center);
\node (c5) at (3, 0) {};
\node (c6) at (3, 1) {\textbullet};
\draw [-biggertip] (c6.center) -- (c5.center);
\node (c7) at (5, 3) {\textbullet};
\node (c8) at (5, 1) {};
\draw [-biggertip] (c7.center) -- (c8.center);

\node (a) at (-1, 3) {\Large$e$};
\draw (0,3) -- ++(7, 0);
\node (b) at (-1, 2) {\Large$x$};
\draw (0,2) -- ++(7, 0);
\node (x) at (-1, 1) {\Large$y$};
\draw (0,1) -- ++(7, 0);
\node (y) at (-1, 0) {\Large$\overline{e}$};
\draw (0,0) -- ++(7, 0);

\end{tikzpicture}
}
\end{center}
\caption{Conditional Comparator Gadget}
\label{fig:cond}
\end{wrapfigure}

The following proposition is a generalization of the corresponding
theorem for Boolean comparator circuits in \cite{Cook12}.

\begin{proposition}
  The language \CCVP{P}{a} is complete under logspace reductions for
  the class \PCC{P}{a} for all finite bounded posets \poset{P} and any
  $a \in \poset{P}$.
\end{proposition}
\begin{proof}
  The problem \CCVP{P}{a} is trivially hard for the class
  \PCC{P}{a}. Let $\lang{L} \in \PCC{P}{a}$ via a logspace-uniform
  circuit family $\{ \circuit{C_n} \}$. Now given $x$ as input, we
  output the tuple $(\circuit{C_n}, x)$ by running the uniformity
  algorithm.

  The fact that $\CCVP{P}{a} \in \PCC{P}{a}$ follows from
  Proposition~\ref{prop:univ-cir}. Given a string, it can be checked
  in logspace whether it is $(n, m)$-valid once $n$ and $m$ are
  fixed. Let $V_{n, m}$ be a logspace uniform comparator circuit over
  the 0--1 lattice that takes an $N$ bit string as input and outputs 1
  iff the input is an $(n, m)$-valid string. Let $N = 2n + m + 2 +
  n(n-1)m$ be the total length of the input. The uniformity machine on
  input $N$ writes out the description of $\bigvee_{{(n, m)}} U_{{n,
      m}} \wedge V_{{n, m}}$ over all $(n, m)$ pairs satisfying $N =
  2n + m + 2 + n(n-1)m$.
\end{proof}

\section{Comparator Circuits over Lattices}
\label{sec:P}

First, we show that comparator circuits over distributive lattices characterize the class $\CC$.

\begin{theorem}
  Let $L$ be any non-trivial finite distributive lattice and $a \in L$
  be an arbitrary element. Then $\CC = \PCC{L}{a}$.
\label{thm:CC}
\end{theorem}
\begin{proof}
  By Birkhoff's representation theorem, every finite distributive
  lattice of $k$ elements is isomorphic to a lattice where each
  element is some subset of $[k]$ (ordered by inclusion) and the join
  and meet operations in the original finite distributive lattice
  correspond to set union and set intersection operations in the new
  lattice. We will use this to simulate a circuit over an arbitrary
  finite distributive lattice $\lattice{L}$ of size $k$ using a
  circuit over the 0--1 lattice. Each line $w$ in the original circuit
  is replaced by $k$ lines $w_1, \ldots, w_k$.  The invariant
  maintained is that whenever a line in the original circuit carries
  $a \in \lattice{L}$, these $k$ lines carry the characteristic vector
  of the set corresponding to the element $a$. Now a gate $(w, x)$ in
  the original circuit is replaced by $k$ gates $(w_1, x_1), \ldots
  ,(w_k, x_k)$ in the new circuit. The correctness follows from the
  fact that meet and join operations in the original circuit
  correspond to set union and set intersection which in turn
  correspond to AND and OR operations of the characteristic vectors.

  We now prove that $\CC \subseteq \PCC{L}{a}$. First, we consider the
  case where $a\neq 0$. We replace the Boolean value 0 by the minimum
  element in $L$ and the Boolean value 1 by the maximum element in
  $L$. The output wire of the new circuit is $o\wedge a$ where $o$ is
  the original output wire. It is easy to see that this circuit output
  $a$ if and only if the original circuit outputs $1$. If $a\neq 0$,
  we can use the fact that the class $\CC$ is closed under
  complementation to construct a Boolean comparator circuit that
  accepts using 0. This can be easily translated to an $\PCC{L}{a}$
  circuit as above.
\end{proof}

Now we consider comparator circuits over fixed finite lattices.  Note
that when characterizing the class \P\ in terms of Boolean circuits,
the fan-out of gates is required to be at least 2. In fact, Mayr and
Subramanian's~\cite{MayrS92} primary motivation while introducing the
class \CC\ was to study fan-out restricted circuits. We show that if
comparator circuits are given the freedom to compute over any
lattice (as opposed to the Boolean lattice on 2 elements), then the
fan-out restriction is irrelevant.

\begin{wrapfigure}{r}{0.4\textwidth}
\centering
\begin{tikzpicture}
 \node (A)  at (1.2, 4) {\smalllabel{$p$}};
 \node (L) at (3.8, 4) {\smalllabel{$0''$}};
 \node (I) at (4.6, 4) {\smalllabel{${0'}^{\circ}$}};
 \node (C) at (5.5, 4) {\smalllabel{$z$}};

 \node (Y) at (2.85, 4.3) {\smalllabel{0}};

 \node (Z) at (3.1, 5.2) {\smalllabel{$1''$}};

 \node (S) at (3.3, 5.8) {\smalllabel{$x$}};

 \node (X)  at (2.85, 6.2) {\smalllabel{1}};
 \node (T)at (4.1, 6.2) {\smalllabel{$0'$}};
 \node (G) at (4.75, 6.2) {\smalllabel{$1'^{\circ}$}};

 \node (B) at (1.2, 6.8) {\smalllabel{$0^{\circ}$}};
 \node (D) at (5.5, 6.8) {\smalllabel{$0'^{\circ\circ}$}};

 \node (P) at (1.25, 7.5) {\smalllabel{$1^\circ$}};
 \node (M) at (2.85, 7.5) {\smalllabel{$w$}};
 \node (U) at (3.55555, 7.5) {\smalllabel{$1'$}};
 \node (E) at (4.5, 7.5) {\smalllabel{$y$}};
 \node (R) at (5.5, 7.5) {\smalllabel{$1'^{\circ\circ}$}};

 \draw (Y) -- (T);
 \draw (X) -- (U);

 \draw (L) -- (Y);
 \draw (Z) -- (X);

 \draw (Y) -- (B);
 \draw (X) -- (P);

 \draw (I) -- (T);
 \draw (G) -- (U);

 \draw (G) -- (R);
 \draw (I) -- (D);

 \draw (X) -- (R);
 \draw (Y) -- (D);

 \draw (Y) -- (X);
 \draw (A) -- (B);
 
 \draw (Y) -- (L);
 
 \draw (L) -- (Z);
 
 \draw (I) -- (G);
 
 \draw (C) -- (D);
 
 \draw (Z) -- (S);
 
 \draw (S) -- (T);
 
 
 \draw (X) -- (M); 
 
 \draw (T) -- (U);

 \draw (G) -- (E);
 
 \draw (B) -- (P); 
 
 \draw (D) -- (R);
 
\end{tikzpicture}
\caption{The poset for simulating \P}
\label{fig:P-poset}
\end{wrapfigure}
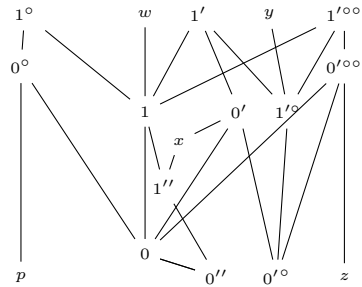

The following lemma describes a fixed finite lattice over which
comparator circuits capture \P. However, it is not clear whether the
class captured by comparator circuits over this lattice is independent
of the accepting element. In Theorem~\ref{thm:P}, we show that there
exists a lattice that captures $\P$ irrespective of the accepting
element. The language $\MCVP$ consists of all tuples $(C, x)$ where
$C$ is a Boolean circuit with only AND, OR and input gates. Here $x
\in {\{0, 1\}}^{n}$ where $n$ is the number of input gates to $C$ and
$x$ specifies the value of each of these input gates. In the proof, we
will reduce in logspace the language $\MCVP$ which is complete for the
class \P\ under logpace reductions to the comparator circuit value
problem over the finite lattice given in Figure~\ref{fig:P-lattice}.

\begin{lemma}
  Let \lattice{L} be the lattice in Figure~\ref{fig:P-lattice}. Then
  $\P = \PCC{L}{1}$ (Note that $1$ is not the maximum element in the
  lattice).
\label{lem:P}
\end{lemma}
\begin{proof}
Let $(\circuit{C}, x)$ be the input to $\MCVP$. For each wire in
\circuit{C}, we add a line to our comparator circuit. The initial
value of the lines that correspond to the input wires of \circuit{C}
are set to 0 or 1 of the poset \poset{P} shown in Figure~\ref{fig:P-poset}
according to whether they are 0 or 1 in $x$. The comparator circuit
simulates \circuit{C} in a level by level fashion maintaining the
invariant that the lines carry 0 or 1 depending on whether they carry
0 or 1 in \circuit{C}. We will show how our comparator circuit
simulates a level 1 OR gate of fan-out 2.  The proof then follows by
an easy induction.

Since $0 \pleq{P} 1$ an AND (OR) gate in \circuit{C} can be simulated
by a meet (join) operation in \poset{P}. The gadget shown in
Figure~\ref{fig:copy} is used to implement the fan-out
operation. The idea is that the first gate in the gadget implements
the AND/OR operation and the rest of the gates in this gadget ``copy''
the result of this operation into the lines $o_1$ and $o_2$ that
correspond to the two output wires of the gate. The reader can
verify that the elements of \poset{P} satisfy the following meet and
join identities. Figure~\ref{fig:copy} shows how one could use the
following identities to copy the output of $a \vee b$ into two lines
(labelled $o_1$ and $o_2$).

The identity $0 \vee 1 = 1$ is used to implement the Boolean AND/OR
operation. This is used by the first gate in Figure~\ref{fig:copy}.
Once the required value is computed. We add a gate between the line
carrying the result of the AND/OR operation and a line with 
value $x$. As the following identities show, this makes two ``copies''
of the result of the Boolean operation.
  $0 \vee x = 0'$, $1 \vee x = 1'$,
  $0 \wedge x = 0''$, $1 \wedge x = 1''$

Now, the following identities can be used to convert the first copy
($0'$ or $1'$) into the original value ($0$ or $1$).
  $0' \wedge y = 0'^{\circ}$, $1' \wedge y = 1'^{\circ}$,
  $0'^{\circ} \vee z = 0'^{\circ\circ}$, $1'^{\circ} \vee z = 1'^{\circ\circ}$,
  $0'^{\circ\circ} \wedge w = 0$, $1'^{\circ\circ} \wedge w = 1$

Similarly, the following identities can be used to convert the second copy
($0''$ or $1''$) into the original value ($0$ or $1$).
  $0'' \vee p = 0^{\circ}$, $1'' \vee p = 1^{\circ}$, 
  $0^{\circ} \wedge w = 0$, $1^{\circ} \wedge w = 1$

The lattice in Figure~\ref{fig:P-lattice}~ is simply the
Dedekind-MacNeille completion of \poset{P}. Since the
Dedekind-MacNeille completion preserves all existing meets and joins,
the same computation can also be performed by this lattice.

To see that for any lattice \lattice{L} and any $a \in \lattice{L}$,
$\PCC{L}{a}$ is in $\P$, observe that in poly-time we can evaluate the
$n^{th}$ comparator circuit from the comparator circuit family for the
language in \PCC{L}{a}.
\end{proof}

Lemma~\ref{lem:P} shows that the complexity class captured by the
comparator circuit could change (Assuming $\CC \neq \P$) depending on
the underlying lattice and the accepting element. In the following
theorem, we show that if we consider any partition lattice, say
$\Pi_i$, that embeds $L$ (in Lemma~\ref{lem:P}), then the complexity
class is $\P$ irrespective of the accepting element. We crucially use
the fact that the circuit in the proof of Lemma~\ref{lem:P} outputs
only the elements 0 and 1 in $L$.

\begin{theorem}
  There exists a constant $i$ such that $\XCC{\Pi_{i}} = \P$.
  \label{thm:P}
\end{theorem}
\begin{proof}
  We know that there exists a finite lattice $L$ and an $a, b \in L$
  such that for any language $\lang{M} \in \P$ there exists a
  comparator circuit family over $L$ that decides $\lang{M}$ by using
  $a$ to accept and $b$ to reject. Also $b < a$.  By
  Pudl\'ak-T{\r{u}}ma theorem \cite{pudlak80}, we know that there
  exists a constant $i$ such that $L$ can be embedded in $\Pi_{i}$. It
  remains to show that the accepting element used does not change the
  complexity. In fact, we will show that for any $X$, $Y \in \Pi_{i}$
  where $X \neq Y$, we can design a comparator circuit family over
  $\Pi_{i}$ that accepts \lang{M} using $X$ and rejects using $Y$. Let
  $A$ and $B$ be the elements in $\Pi_{i}$ that $a$ and $b$ gets
  mapped to by this embedding ($B < A$).  Then there exists a circuit
  family $C$ over $\Pi_{i}$, deciding \lang{M}, that accepts using $A$
  and rejects using $B$. We will construct a circuit family $C'$ over
  $\Pi_{i}$ from $C$ such that $C'$ uses $1$ to accept and $0$ to
  reject. Here 1 and 0 are the maximum and minimum elements in
  $\Pi_i$. Now if we let $x$ be the output of a circuit in the circuit
  family $C$, we can construct $C'$ by computing $\checkgeb{A}(x)$
  (See Proposition~\ref{prop:checkgeb}). Similarly, we can construct a
  circuit family $C''$ that accepts using 0 and rejects using 1 by
  reducing the language \lang{M} to $\overline{\MCVP}$ and then
  applying the construction in Lemma~\ref{lem:P} and then computing
  $\checkgeb{A}(x)$ on the output of this circuit. The required
  circuit family is then the one computing $(X \wedge C') \vee (Y
  \wedge C'')$.
\end{proof}

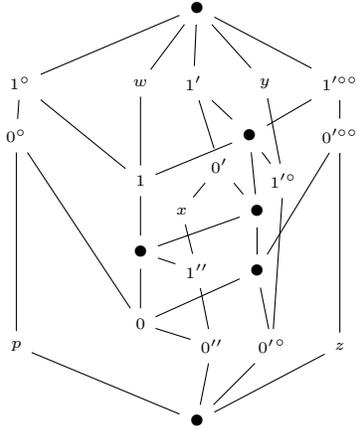
\begin{figure}
\centering
\begin{tikzpicture}
 \node (K)  at (3.6, 3) {\textbullet};

 \node (A)  at (1.2, 4) {\smalllabel{$p$}};
 \node (L) at (3.8, 4) {\smalllabel{$0''$}};
 \node (I) at (4.6, 4) {\smalllabel{${0'}^{\circ}$}};
 \node (C) at (5.5, 4) {\smalllabel{$z$}};

 \node (Y)at (2.85, 4.3) {\smalllabel{0}};

 \node (Z) at (3.6, 5) {\smalllabel{$1''$}};
 \node (J) at (4.4, 5) {\textbullet};

 \node (W) at (2.85, 5.25) {\textbullet};

 \node (S) at (3.4, 5.8) {\smalllabel{$x$}};
 \node (V) at (4.4, 5.8) {\textbullet};

 \node (X)  at (2.85, 6.2) {\smalllabel{1}};
 \node (T)at (3.9, 6.4) {\smalllabel{$0'$}};
 \node (G) at (4.75, 6.2) {\smalllabel{$1'^{\circ}$}};

 \node (B) at (1.2, 6.8) {\smalllabel{$0^{\circ}$}};
 \node (H) at (4.3, 6.8) {\textbullet};
 \node (D) at (5.5, 6.8) {\smalllabel{$0'^{\circ\circ}$}};

 \node (P) at (1.25, 7.5) {\smalllabel{$1^\circ$}};
 \node (M) at (2.85, 7.5) {\smalllabel{$w$}};
 \node (U) at (3.55555, 7.5) {\smalllabel{$1'$}};
 \node (E) at (4.5, 7.5) {\smalllabel{$y$}};
 \node (R) at (5.5, 7.5) {\smalllabel{$1'^{\circ\circ}$}};

 \node (F) at (3.6, 8.5) {\textbullet};

 \draw (Y) -- (B);
 \draw (X) -- (P);

 \draw (K) -- (A);
 \draw (K) -- (L);
 \draw (K) -- (I);
 \draw (K) -- (C);

 \draw (A) -- (B);

 \draw (Y) -- (W);
 \draw (Y) -- (J);
\draw (Y) -- (L);

 \draw (L) -- (Z);

 \draw (I) -- (J);
\draw (I) -- (G);

 \draw (C) -- (D);

 \draw (Z) -- (W);
\draw (Z) -- (S);

 \draw (J) -- (V);
 \draw (J) -- (D);

\draw (W) -- (X);
 \draw (W) -- (V);

 \draw (S) -- (T);

\draw (V) -- (T);
 \draw (V) -- (H);


 \draw (X) -- (H);
 \draw (X) -- (M); 

\draw (T) -- (U);

\draw (G) -- (H);
 \draw (G) -- (E);

 \draw (B) -- (P); 


 \draw (H) -- (U);
 \draw (H) -- (R); 

\draw (D) -- (R);

\draw (P) -- (F);
 \draw (M) -- (F);
\draw (U) -- (F);
 \draw (E) -- (F);
 \draw (R) -- (F);

\end{tikzpicture}
\caption{The lattice for simulating \P}
\label{fig:P-lattice}
\end{figure}

\begin{figure}
\centering
\scalebox{0.5}{%
\begin{tikzpicture}
\node (c1) at (1, 7) {\textbullet};
\node (c2) at (1, 6) {};
\draw [-biggertip] (c1.center) -- (c2.center);
\node (c3) at (3, 6) {\textbullet};
\node (c4) at (3, 5) {};
\draw [-biggertip] (c3.center) -- (c4.center);
\node (c5) at (5, 5) {};
\node (c6) at (5, 4) {\textbullet};
\draw [-biggertip] (c6.center) -- (c5.center);
\node (c7) at (7, 4) {\textbullet};
\node (c8) at (7, 3) {};
\draw [-biggertip] (c7.center) -- (c8.center);
\node (c9) at (9, 3) {};
\node (c10) at (9, 2) {\textbullet};
\draw [-biggertip] (c10.center) -- (c9.center);
\node (c11) at (11, 6) {\textbullet};
\node (c12) at (11, 1) {};
\draw [-biggertip] (c11.center) -- (c12.center);
\node (c13) at (13, 1) {};
\node (c14) at (13, 0) {\textbullet};
\draw [-biggertip] (c14.center) -- (c13.center);

\node (a) at (-1, 7) {\Large$(a)$};
\draw (0,7) -- ++(15, 0);
\node (b) at (-1, 6) {\Large$(b)$};
\draw (0,6) -- ++(15, 0);
\node (x) at (-1, 5) {\Large$x$};
\draw (0,5) -- ++(15, 0);
\node (y) at (-1, 4) {\Large$y$};
\draw (0,4) -- ++(15, 0);
\node (z) at (-1, 3) {\Large$z$};
\draw (0,3) -- ++(15, 0);
\node (w1) at (-1, 2) {\Large$(o_1) w$};
\draw (0,2) -- ++(15, 0);
\node (p) at (-1, 1) {\Large$p$};
\draw (0,1) -- ++(15, 0);
\node (w2) at (-1, 0) {\Large$(o_2) w$};
\draw (0,0) -- ++(15, 0);

\end{tikzpicture}
}
\caption{Copy $a \vee b$ into $o_1$ and $o_2$}
\label{fig:copy}
\end{figure}
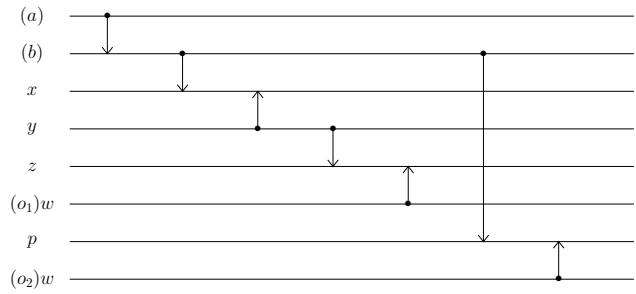

If we can show that there exists a finite distributive lattice such
that the poset in Figure~\ref{fig:P-poset} can be embedded in that
lattice while preserving all existing meets and joins, then $\P =
\CC$. In the following theorem, we show that such an embedding is not
possible.

\begin{theorem}
  The poset in Figure~\ref{fig:P-poset} cannot be embedded into any
  distributive lattice while preserving all meets and joins.
  \label{thm:impossibility}
\end{theorem}
\begin{proof}
  We use proof by contradiction. Assume that such an embedding
  exists. Then by Birkhoff's representation theorem, the elements of
  the poset in Figure~\ref{fig:P-poset} can be labelled by finite sets
  such that lub and glb operations over the embedding distributive
  lattice correspond to union and intersection of these sets
  respectively. We will denote the set labelling each element of the
  poset in Figure~\ref{fig:P-poset} by the corresponding uppercase
  letter except that 0, 1, $0'$, $1'$, $0''$ and $1''$ are labelled by
  $A$, $B$, $A'$, $B'$, $A''$ and $B''$ respectively.

  Let $\{ x_{1}, \ldots ,x_{k} \} = B\setminus A$. Since $A \subset
  B$, we have $k \geq 1$. Our first goal is to prove that all of these
  $x_{i}$ must be in $B''$ too.  Since $B \subset W$, we have for all
  $i$ that $x_{i} \in W$. Now suppose for contradiction that there
  exists an $i$ such that $x_{i} \in P$, then we can conclude that
  $x_{i} \in A$ since $A = (A'' \cup P) \cap W$. So for all $i$ we
  have $x_{i} \notin P$. Since $B = (B'' \cup P) \cap W$ and $x_{i}
  \notin P$, we have $x_{i} \in B''$. But then for all $i$ we have
  $x_{i} \in A'$ as $A' \supset B''$. So $A' \supseteq B$ which is a
  contradiction since $A'$ and $B$ are incomparable.
\end{proof}


\section{Comparator Circuits over Bounded Posets}
\label{sec:NP}
In this section, we consider the most general form of comparator
circuits, i.e., we consider comparator circuits over fixed finite
bounded posets. We show that the resulting complexity class is 
the class \NP.

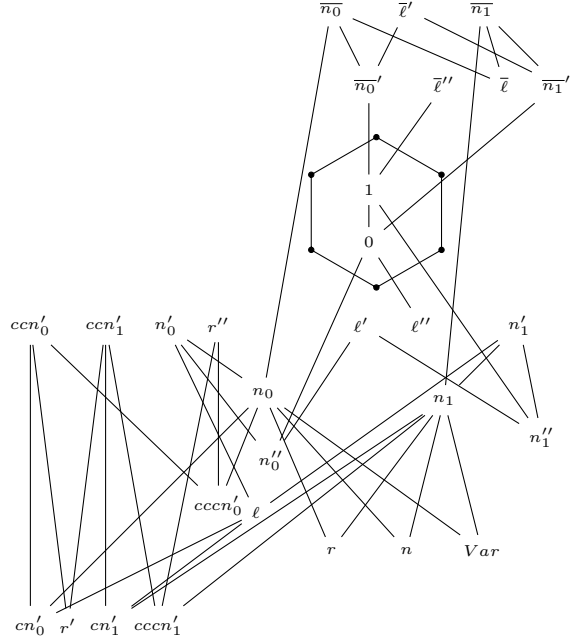
\begin{figure}
\centering
\begin{tikzpicture}
 \node (Y)at (5.5, 8.6) {\smalllabel{0}};
 \node (X)  at (5.5, 9.3) {\smalllabel{1}};

 \draw (Y) -- (X);

 \begin{scope}[shift={(5.6,8)}]
   \draw[fill,rotate=30]
   \foreach \a in {0,60,...,300} { 
     -- ++(\a:1cm) circle (1pt)

   }-- (0:1cm);
 \end{scope}

 \node (NA) at (1, 7.5) {\smalllabel{$ccn_0'$}};
 \node (NB) at (3.5, 7.5) {\smalllabel{$r''$}};
 \node (NC) at (1, 3.5) {\smalllabel{$cn_0'$}};
 \node (ND) at (1.5, 3.5) {\smalllabel{$r'$}};
 \node (NE) at (2, 3.5) {\smalllabel{$cn_1'$}};
 \node (NF) at (2.7, 3.5) {\smalllabel{$cccn_1'$}};
 \node (NG) at (2, 7.5) {\smalllabel{$ccn_1'$}};
 \node (NH) at (3.5, 5.1) {\smalllabel{$cccn_0'$}};
 \node (NI) at (4.1, 6.6) {\smalllabel{$n_0$}};
 \node (NJ) at (2.8, 7.5) {\smalllabel{$n_0'$}};
 \node (NK) at (5.5, 10.7) {\smalllabel{$\overline{n_0}'$}};
 \node (NL) at (5, 11.7) {\smalllabel{$\overline{n_0}$}};
 \node (NM) at (6, 11.7) {\smalllabel{$\overline{\ell}'$}};
 \node (NN) at (7, 11.7) {\smalllabel{$\overline{n_1}$}};
 \node (NO) at (8, 10.7) {\smalllabel{$\overline{n_1}'$}};
 \node (NP) at (6.5, 10.7) {\smalllabel{$\overline{\ell}''$}};
 \node (NQ) at (7.3, 10.7) {\smalllabel{$\overline{\ell}$}};
 \node (NR) at (6.2, 7.5) {\smalllabel{$\ell''$}};
 \node (NS) at (5.4, 7.5) {\smalllabel{$\ell'$}};
 \node (NT) at (4.2, 5.7) {\smalllabel{$n_0''$}};
 \node (NU) at (5, 4.5) {\smalllabel{$r$}};
 \node (NV) at (7.8, 6) {\smalllabel{$n_1''$}};
 \node (NW) at (7.5, 7.5) {\smalllabel{$n_1'$}};
 \node (NX) at (6.5, 6.5) {\smalllabel{$n_1$}};
 \node (NY) at (7, 4.5) {\smalllabel{$Var$}};
 \node (NZ) at (6, 4.5) {\smalllabel{$n$}};
 \node (NAA) at (4, 5) {\smalllabel{$\ell$}};

 \draw (NA) -- (NC);
 \draw (NA) -- (ND);
 \draw (NA) -- (NH);

 \draw (NB) -- (NH);
 \draw (NB) -- (NF);

 \draw (NC) -- (NI);

 \draw (ND) -- (NG);

 \draw (NE) -- (NAA);
 \draw (NE) -- (NG);
 \draw (NE) -- (NX);

 \draw (NF) -- (NG);
 \draw (NF) -- (NX);

 \draw (NH) -- (NI);

 \draw (NI) -- (NJ);
 \draw (NI) -- (NU);
 \draw (NI) -- (NZ); 
 \draw (NI) -- (NY);
 \draw (NI) -- (NL);
 
 \draw (NJ) -- (NT);
 \draw (NJ) -- (NAA);

 \draw (NK) -- (NL);
 \draw (NK) -- (NM);
 \draw (NK) -- (X);

 \draw (NL) -- (NQ);

 \draw (NM) -- (NO);

 \draw (NN) -- (NO);
 \draw (NN) -- (NQ);

 \draw (NO) -- (Y);

 \draw (NP) -- (X);

 \draw (NR) -- (Y); 

 \draw (NS) -- (NV);

 \draw (NT) -- (Y);
 
 \draw (NU) -- (NX); 

 \draw (NV) -- (NW);
 \draw (NV) -- (X);

 \draw (NW) -- (NX);

 \draw (NX) -- (NY);
 \draw (NX) -- (NZ);

 \draw (NT) -- (NS);
 \draw (NC) -- (NAA);
 \draw (NAA) -- (NW);
 \draw (NX) -- (NN);

\end{tikzpicture}
\caption{The poset for simulating \NP}
\label{fig:NP-poset}
\end{figure}

\begin{theorem}
  Let $P$ be any poset and let $a \in P$ be an arbitrary element in
  $P$, then $\PCC{P}{a} \subseteq \NP$. Also, there exists a finite
  poset $P$ and an $a \in P$ such that $\NP = \PCC{P}{a}$.
\label{thm:NP}
\end{theorem}
\begin{proof}
  First, we prove that there exists a poset \poset{P} and an accepting
  element $a \in \poset{P}$ such that $\NP \subseteq \PCC{P}{a}$. Let
  \poset{P} be the poset in Figure~\ref{fig:NP-poset}. We will reduce
  the well-known \NP-complete problem \SAT\ into \CCVP{P}{a}. Without loss of generality, we can assume that the circuit does not contain any NOT gates.

  Note that the poset \poset{P} contains the poset in the proof of
  Theorem~\ref{thm:P}. This is represented by the hexagon in
  Figure~\ref{fig:NP-poset}. The elements marked 0 and 1 inside this
  hexagon are the elements marked 0 and 1 in Figure~\ref{fig:P-poset}.
  This containment ensures that we can implement all operations that
  we used while simulating \MCVP\ to be used here as well.  Let
  \circuit{C} be the input to the \SAT\ problem. The 0 and 1 values
  carried by wires will be represented by 0 and 1 in \poset{P} as in
  the proof of Theorem~\ref{thm:P}. The non-trivial part is to
  simulate the input variables $x_1, \ldots , x_n$. These input
  variables to \circuit{C} are handled by non-deterministically
  generating 0 or 1 (of \poset{P}) on the lines corresponding to the
  wires attached to these input gates. We also have to ensure that
  when we non-deterministically generate the values of input variables,
  the values generated for $x_i$ and $\overline{x_i}$ are
  consistent. This is ensured by generating $x_i$ non-deterministically
  and then complementing the generated value to get
  $\overline{x_i}$. The fan-out operation is implemented as in the
  proof of Theorem~\ref{thm:P}.

  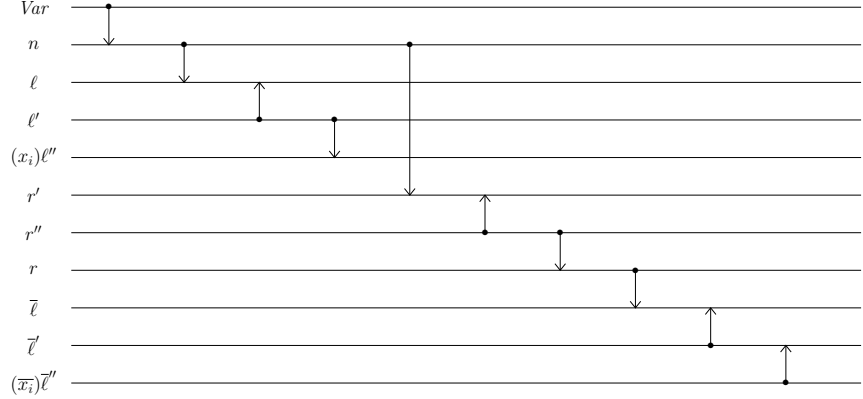
\begin{figure}
\centering
\scalebox{0.5}{%
\begin{tikzpicture}
\node (c1) at (1, 10) {\textbullet};
\node (c2) at (1, 9) {};
\draw [-biggertip] (c1.center) -- (c2.center);
\node (c3) at (3, 9) {\textbullet};
\node (c4) at (3, 8) {};
\draw [-biggertip] (c3.center) -- (c4.center);
\node (c5) at (5, 8) {};
\node (c6) at (5, 7) {\textbullet};
\draw [-biggertip] (c6.center) -- (c5.center);
\node (c7) at (7, 6) {};
\node (c8) at (7, 7) {\textbullet};
\draw [-biggertip] (c8.center) -- (c7.center);
\node (c9) at (9, 5) {};
\node (c10) at (9, 9) {\textbullet};
\draw [-biggertip] (c10.center) -- (c9.center);
\node (c11) at (11, 4) {\textbullet};
\node (c12) at (11, 5) {};
\draw [-biggertip] (c11.center) -- (c12.center);
\node (c13) at (13, 3) {};
\node (c14) at (13, 4) {\textbullet};
\draw [-biggertip] (c14.center) -- (c13.center);
\node (c15) at (15, 2) {};
\node (c16) at (15, 3) {\textbullet};
\draw [-biggertip] (c16.center) -- (c15.center);
\node (c17) at (17, 2) {};
\node (c18) at (17, 1) {\textbullet};
\draw [-biggertip] (c18.center) -- (c17.center);
\node (c19) at (19, 1) {};
\node (c20) at (19, 0) {\textbullet};
\draw [-biggertip] (c20.center) -- (c19.center);

\node (var) at (-1, 10) {\Large$\var{Var}$};
\draw (0,10) -- ++(21, 0);
\node (n) at (-1, 9) {\Large$n$};
\draw (0,9) -- ++(21, 0);
\node (L) at (-1, 8) {\Large$\ell$};
\draw (0,8) -- ++(21, 0);
\node (LP) at (-1, 7) {\Large$\ell'$};
\draw (0,7) -- ++(21, 0);
\node (LPP) at (-1, 6) {\Large$(x_i) \ell''$};
\draw (0,6) -- ++(21, 0);
\node (RP) at (-1, 5) {\Large$r'$};
\draw (0,5) -- ++(21, 0);
\node (RPP) at (-1, 4) {\Large$r''$};
\draw (0,4) -- ++(21, 0);
\node (R) at (-1, 3) {\Large$r$};
\draw (0,3) -- ++(21, 0);
\node (LB) at (-1, 2) {\Large$\overline{\ell}$};
\draw (0,2) -- ++(21, 0);
\node (LBP) at (-1, 1) {\Large$\overline{\ell}'$};
\draw (0,1) -- ++(21, 0);
\node (LBPP) at (-1, 0) {\Large$(\overline{x_i}) \overline{\ell}''$};
\draw (0,0) -- ++(21, 0);

\end{tikzpicture}
}
\caption{Nondeterministically generate $x_i$ and $\overline{x_i}$}
\label{fig:nondet}
\end{figure}

  Note that the minimal upper bounds for the elements $\mathit{Var}$
  and $n$ in the poset \poset{P} are $n_0$ and $n_1$. These values
  stand for a non-deterministically generated 0 and 1 resp. Now for
  each variable $x_i$ we take the minimal upper bound of these two
  elements in \poset{P} to non-deterministically generate the value of
  $x_i$. The only thing that remains to be done is to make the
  corresponding $\bar{x_i}$ variable consistent, i.e., when a 0 is
  generated non-deterministically for $x_i$, we have to ensure that all
  lines carrying $\bar{x_i}$ in that non-deterministic path carry the
  value 0. The sequence of meet and join identities that we are going
  to describe can be used to implement this
  computation. Figure~\ref{fig:nondet} shows how to generate $x_i$ and
  $\overline{x_i}$ consistently in a non-deterministic fashion using
  the identities given below.

  The following identity enables us to non-deterministically generate a
  0 or a 1. Note that we are only generating $n_0$ and $n_1$ at this
  point. But we will later convert this into 0 or 1 that are used for
  implementing the Boolean operations.

  \begin{align*}
    \var{Var} \vee n &= \{ n_0, n_1 \}
  \end{align*}
  
  Now we use the following identities to convert $n_0$ or $n_1$ into a 0
  or a 1 respectively.
  \begin{align*}
    \ell \vee n_0 &= n'_{0} & \ell \vee n_1 &= n'_{1}\\ 
    \ell' \wedge n'_{0} &= n''_{0} & \ell' \wedge n'_{1} &= n''_{1}\\ 
    \ell'' \vee n''_{0} &= 0 & \ell'' \vee n''_{1} &= 1\\  
  \end{align*}
  
  Note that the original $n_0$ or $n_1$ that was generated will be
  destroyed by the above sequence of operations (By doing $\ell \wedge
  n_0$ for ex.). Using the following identities, we ensure that the
  original value generated non-deterministically is restored.
  \begin{align*}
    \ell \wedge n_0 &= \var{cn_0^{'}} & \ell \wedge n_1 &= \var{cn_1^{'}}\\ 
    r' \vee \var{cn_0{'}} &= \var{ccn_0{'}} & r' \vee \var{cn_1{'}} &= \var{ccn_1{'}}\\ 
    r'' \wedge \var{ccn_0{'}} &= \var{cccn_0{'}} & r'' \wedge \var{ccn_1{'}} &= \var{cccn_1{'}}\\ 
    r \vee \var{cccn_0{'}} &= n_0 & r \vee \var{cccn_1{'}} &= n_1\\ 
  \end{align*}
  
  Now we use the restored value along with the following identities to
  generate the value for the line carrying $\overline{x_i}$.
  \begin{align*}
    \overline{\ell} \vee n_0 &= \overline{n_{0}} & \overline{\ell} \vee n_1 &= \overline{n_{1}}\\ 
    \overline{\ell}' \wedge \overline{n_{0}} &= \overline{n_{0}}' & \overline{\ell}' \wedge \overline{n_{1}} &= \overline{n_{1}}'\\ 
    \overline{\ell}'' \wedge \overline{n_{0}}' &= 1 & \overline{\ell}'' \wedge \overline{n_{1}}' &= 0\\
  \end{align*}

  The reduction from \SAT{} is as follows. First, we use the reduction
  from \CVP{} to \PCC{P}{a} to construct a \PCC{P}{a} circuit, say $C$,
  that evaluates the input formula. Then, we construct a circuit for
  non-deterministically generating 0/1 values for all the variables in
  the formula. The wires of this circuit that carry the generated
  values are then connected to the input wires in $C$ that take the
  values of variables in the formula as input. It is easy to see that
  the resulting circuit evaluates to 1 if and only if the formula is
  satisfiable.

  To see that \PCC{P}{a} is in \NP{}, observe that we can evaluate any
  \PCC{P}{a} circuit in \NP{} by guessing the output value of each
  gate to be one of the possible values. i.e., if the gate is an OR
  (AND) gate taking $a$ and $b$ as input, we non-deterministically
  guess that the gate outputs one of the values in $a\vee b$ ($a\wedge
  b$). Finally, we simply check whether the value on the output wire
  is in the accepting set.
\end{proof}
\section{Skew Comparator Circuits}

In this section, we study the skew comparator circuits defined in the
preliminaries.  We show that \skewcc\ is the class \L. Recall that the
class \NL\ can be characterized as the set of all languages computed
by logpsace-uniform Boolean circuits with skewed AND gates. So the
result in this section draws a parallel between the $\P$ vs $\CC$
problem and the $\NL$ vs $\L$ problem. It immediately follows that
\skewcc\ over distributive lattices also characterize the class $\L$.

We begin by considering a canonical complete problem for the class
$\L$.
The language \dgap\ consists of all tuples $(G, s, t)$ where $G = (V,
E)$ is a directed graph where each vertex has out-degree at most one
and $s, t \in V$ and there is a directed path from $s$ to $t$.
%
%
We use a variant of $\dgap$ problem in our setting. The variant
(called $\dgapp$) is that the out-degree constraint is not applied to
$s$.
It is easy to see that $\dgapp$ is also in $\L$. Indeed, for each
neighbour $u$ of $s$, run the \dgap\ algorithm to check whether $t$ is
reachable from $u$.

\begin{theorem}
  $\skewcc = \L$
  \label{thm:L}
\end{theorem}

\begin{proof}
  ($\subseteq$) Let $\lang{L} \in \skewcc$. We will prove that
  $\lang{L} \in \L$ by reducing \lang{L} to \dgapp. The reduction is
  as follows. Observe that we can reduce the language \lang{L} to
  \skewccvp\ by a logspace reduction (using the uniformity
  algorithm). Then we reduce \skewccvp\ to \dgapp\ as follows. Let $C$
  be an instance of \skewccvp. For each wire in $C$ add a vertex to
  the graph $G$. The vertex corresponding to the output wire is the
  destination vertex $t$. Add a source vertex $s$. The edges of $G$
  are as follows. For each vertex $v$ that corresponds to an input
  wire of $C$ having value $1$, add the edge $(s, v)$ to the
  graph. Now consider a comparator gate $g$ in $C$ with input wires
  $e_{1}$ and $e_{2}$ and AND output wire $e_{3}$ and OR output wire
  $e_{4}$. There are two cases.
  
  \begin{description}
  \item [Gate $g$ has an AND output] Without loss of generality, assume that $e_2$ is an input wire
    to $C$. If $e_2 = 1$, then add the edges $(e_{1}, e_{3})$ and
    $(e_{2}, e_{4})$ to $G$. If $e_{2} = 0$, then add the edge $(e_{1},
    e_{4})$ to the graph $G$.
    
  \item [Gate $g$ has an unused AND output] Add the edges $(e_{1},
    e_{4})$ and $(e_{2}, e_{4})$. Note that it is easy to check in
    logspace whether the AND output of a gate is used or not. Simply
    scan forward on the input to check whether any gate in the input
    after $g$ is incident on the AND output line of $g$ or not.
  \end{description}
  
  It is clear that $G$ has an $s$--$t$ path if and only if $C$ outputs $1$.  This
  follows from the observation that every vertex $v$ in $G$ where $v
  \neq s$ corresponds to a wire in $C$ and $v$ is reachable from $s$
  if and only if the wire corresponding to $v$ carries the value $1$.  All
  vertices other than $s$ in $G$ have out-degree at most
  $1$. Furthermore, the reduction can be implemented in logspace.

  ($\supseteq$) Let $\lang{L} \in \L$ and let $B$ be a poly-sized
  layered branching program deciding \lang{L}.  We will design a skew
  comparator circuit $C$ to simulate $B$.  Let $s$ be a state in $B$
  reading $x_{i}$ and let the edge labelled $1$ be directed towards a
  state $t$ and let the edge labelled $0$ be directed towards a state
  $u$. Then the gadget shown in Figure~\ref{fig:bp-to-cc} simulates
  this part of the BP $B$ (We say that this gadget corresponds to the
  state $s$). The truth table for this gadget is shown in the
  Table~\ref{tab:bp-to-cc}. This table assumes that the lines $t$ and
  $u$ carry the value $0$ initially. The value of the line labelled
  $s$ will be $1$ on input $x$ just before the gates in this gadget
  are evaluated if and only if the input $x$ reaches the state $s$ in
  $B$. It is clear that after all the gates in this gadget are
  evaluated, the value of the line labelled $t$ (or $u$) is $1$ if and
  only if the input $x$ reaches $t$ (or $u$ resp.) in $B$. 

  Now the circuit $C$ is as follows. For each state in $B$ introduce a line in
  $C$ and for each state in each layer from the first layer to the
  last layer, in that order, add the gates in the gadgets
  corresponding to these states in the same order to $C$. Note that
  the lines annotated $x_{i}$ and $\overline{x_{i}}$ in a gadget are
  only used in that gadget. When these values are required again, new
  annotated lines are used. The line corresponding to the accepting
  state is the output line. The initial value of lines corresponding
  to each state other than the start state of $B$ is $0$ and the
  initial value of the line corresponding to the start state is
  $1$. Also the circuit is a skew circuit since all used AND gates in
  the gadget are skew.  For establishing the correctness, we observe
  that the following claim holds. The circuit $C$ outputs 1 on input
  $x$ if and only if there is a path in $B$ from the start state to
  the accepting state on input $x$.
To complete the correctness proof, we prove the following claim:
  \begin{claim}
    The circuit $C$ outputs 1 on input $x$ if and only if there is a path in $B$
    from the start state to the accepting state on input $x$.
  \end{claim}
  \begin{proof}
    Let the $i^{\text{th}}$ \emph{block} of $C$ include all the
    gadgets corresponding to all the states in layer $i$ of $B$. We
    will prove the more general claim that after all gates up to and
    including the $i^{\text{th}}$ block are evaluated, if we consider
    all the lines that correspond to states in the
    ${(i+1)}^{\text{th}}$ layer of $B$, the only line that will have a
    value 1 will correspond to the state on ${(i+1)}^{\text{th}}$
    layer reached on input $x$. We will prove this by induction on the
    layer number.

    \paragraph{Base case: $i = 0$} Since there is only the start state
    in layer $1$ and it is initialized to the value 1, the base case
    is true.

    \paragraph{Induction} Assume that the claim is true for $i$. Let
    $s$ be the state in the ${(i+1)}^{\text{th}}$ layer that is
    reached by $x$ and let $t$ be the state in the
    ${(i+2)}^{\text{th}}$ that is reached by $x$. Now from the truth
    table in Table~\ref{tab:bp-to-cc} it is clear that after the
    gadget for state $s$ is evaluated the value of line $t$ will
    become $1$. Also, from the truth table, it is clear that the
    values of all the other lines that correspond to states in the
    ${(i+2)}^{\text{th}}$ layer remains $0$. Notice that all gates in
    block $i+1$ incident on $t$ are OR gates. So once the value of
    line $t$ becomes $1$, it remains so until block $i+2$.
  \end{proof}
  Let $s$ be the number of states in $B$. Then the number of lines in
  $C$ is at most $3s$ and the number of gates in $C$ is at most
  $4s$. Since $B$ is poly-size, so is $C$.
\begin{figure}[h]
  \begin{subfigure}[b]{0.4\textwidth}
    \begin{center}
      \begin{tabular}{ccc|cc}
        $s$ & $x_{i}$ & $\overline{x_{i}}$ & $t$ & $u$ \\ \hline
        0 & 1 & 0 & 0 & 0 \\
        0 & 0 & 1 & 0 & 0 \\
        1 & 1 & 0 & 1 & 0 \\
        1 & 0 & 1 & 0 & 1 \\
      \end{tabular}
    \end{center}
    \caption{Truth Table for the gadget for BPs}
    \label{tab:bp-to-cc}
  \end{subfigure}
  \hfill
  \begin{subfigure}[b]{0.4\textwidth}
    \begin{center}
      \scalebox{0.5}{%
        \begin{tikzpicture}
          \node (c1) at (1, 3) {\textbullet};
          \node (c2) at (1, 2) {};
          \draw [-biggertip] (c1.center) -- (c2.center);
          \node (c3) at (3, 3) {\textbullet};
          \node (c4) at (3, 0) {};
          \draw [-biggertip] (c3.center) -- (c4.center);
          \node (c7) at (5, 2) {\textbullet};
          \node (c8) at (5, 1) {};
          \draw [-biggertip] (c7.center) -- (c8.center);
          \node (c9) at (7, 2) {\textbullet};
          \node (c10) at (7, -1) {};
          \draw [-biggertip] (c9.center) -- (c10.center);

          \node (s) at (-1, 3) {\Large$s$};
          \draw (0,3) -- ++(9, 0);
          \node (xi) at (-1, 2) {\Large$x_{i}$};
          \draw (0,2) -- ++(9, 0);
          \node (xib) at (-1, 1) {\Large$\overline{x_{i}}$};
          \draw (0,1) -- ++(9, 0);
          \node (t) at (-1, 0) {\Large$t$};
          \draw (0,0) -- ++(9, 0);
          \node (u) at (-1, -1) {\Large$u$};
          \draw (0,-1) -- ++(9, 0);
        \end{tikzpicture}
      }
      \caption{The gadget for simulating BPs}
      \label{fig:bp-to-cc}
    \end{center}
  \end{subfigure}
\end{figure}

It is easy to see that this reduction can be implemented in $\NC^1$.
\end{proof}
Since the construction in Theorem~\ref{thm:CC} preserves skewness of
the circuit, we have the following corollary.
\begin{corollary}
Let $L$ be any distributive lattice and let $a$ be any element in $L$, then $\PskewCC{L}{a} = \L$.
\end{corollary}

We now look at skewed comparator circuits over arbitrary lattices and
show that they also characterize the class \P{}. We prove this by
modifying the proof of Theorem~\ref{thm:P}. More specifically, we show
that by changing the underlying lattice, we can simulate any AND gate
using OR gates and skewed AND gates.

\begin{theorem}
There exists an $i$ such that $\XSkewCC{\Pi_i} = \P$.
\label{thm:SkewP}
\end{theorem}
\begin{proof}
  We will start with the comparator circuit in the proof of
  Theorem~\ref{thm:P} and show how to replace AND gates in that
  circuit with OR gates and skewed AND gates. We start with the poset
  shown in Figure~\ref{fig:P-poset}. We then add new elements $q$,
  $r$, and $v$ to the poset that satisfies the following relations.

  \begin{align*}
    q \wedge 0 &= 0^q\\
    q \wedge 1 &= 1^q\\
    r \vee 0^q &= 0^r\\
    r \vee 1^q &= 1^r\\
    v \wedge 0^r &= 0+\\
    v \wedge 1^r &= 1+
  \end{align*}

  Here, the elements $0+$ and $1+$ can be thought of as placeholders
  for $0$ and $1$ respectively. We then add four new elements to the
  poset $(a, b)+$ where $a, b \in \{0, 1\}$ satisfying ${a+} \vee b =
  (a, b)+$. Then we introduce new elements $s$, $t$, and $u$ such that

  \begin{align*}
    s \wedge (1, 1)+ &= 1^s\\
    s \wedge (a, b)+ &= 0^s\text{, otherwise}\\
    t \vee 0^s &= 1^t\\
    t \vee 1^s &= 0^t\\
    u \wedge 0^t &= 0\\
    u \wedge 1^t &= 1
  \end{align*}

  Now given an AND gate computing $x \land y \in \{0, 1\}$ in the circuit in
  the proof of Theorem~\ref{thm:P} (Note that the non-skew AND gates
  in that circuit always take input from $\{0, 1\}$). We replace that
  AND gate with the following sequence of operations. First we compute
  $((x \wedge q) \vee r) \wedge v$ to yield $x+$. We then OR the wires
  containing $x+$ and $y$ (This is the only non-skew gate used in this
  construction) to yield $(x, y)+$. Finally, we compute $(({(x, y)+}
  \wedge\ s) \vee t) \wedge u$ to yield the required value $x\wedge
  y$. Note that all AND gates used in this construction are
  skewed. The complete set of relations added to the poset in
  Figure~\ref{fig:P-poset} is listed in Figure~\ref{fig:P-skew-poset}.

  We use the same argument as in the proof of Theorem~\ref{thm:P} to
  show that this can be simulated in a partition lattice irrespective
  of the accepting element.
\end{proof}

\begin{figure}[ht]
\begin{align*}
  t &\leq 0^{t} & 0^{t} &\leq 1^{t} & 0^{s} &\leq 0^{t} & 0^{s} &\leq 1^{s} & 0^{s} &\leq {(0,0)+}\\
  1^{s} &\leq 1^{t} & 1^{s} &\leq s & 1^{s} &\leq {(1,1)+} & {(0, 1)+} &\leq {(1,1)+} & {(1,0)+} &\leq {(1,1)+}\\
  {(0,0)+} &\leq {(0,1)+} & {(0,0)+} &\leq {(1,0)+} & {0+} &\leq {(0,0)+} & {0+} &\leq {1+} & {0+} &\leq 0^{r}\\
  {1+} &\leq {(0,1)+} & {1+} &\leq v & {1+} &\leq 1^{r} & 0^{r} &\leq 1^{r} & r &\leq 0^{r}\\
  {0^{q}} &\leq {0^{r}} & {0^{q}} &\leq 1^{q} & {0^{q}} &\leq 0 & 1^{q} &\leq 1^{r} & 1^{q} &\leq q\\
  {1^{q}} &\leq 1 & 1 &\leq u & 1 &\leq 1^{t} & 1 &\leq {(1,0)+} & 0 &\leq 0^{t}\\
  0 &\leq {(0,0)+}
\end{align*}
  \caption{Relations added to the poset in Figure~\ref{fig:P-poset} to make the circuit skewed}
  \label{fig:P-skew-poset}
\end{figure}

\section{Formulae over Lattices}

It is well known that languages decided by poly-size formulae is the
class $\NC^{1}$. By definition, the class $\NC^{1}$ is also the class
of languages decided by log-depth Boolean circuits with bounded fan-in
AND and OR gates.  We can modify Definition~\ref{def:comp} to define
formulae over finite bounded posets. We denote by \formula{L}{a},
where $L$ is a lattice and $a \in L$, the class of all languages
decided by poly-size formulae over $L$ using $a$ as the accepting
element. In this section, we show that the languages computed by
poly-size formulae over any fixed finite lattice is the class
$\NC^{1}$. The proof for the Boolean case is by \cite{Spira71} and it
works by depth reducing an arbitrary formula of poly-size to a Boolean
formula of poly size and log depth. The depth reduction is done by
identifying a separator vertex in the tree and then evaluating the
separated components (which are smaller circuits) in parallel. We show
that a similar argument can be extended to the case of finite lattices
as well. Our main theorem in this section is the following.

\begin{theorem}
  Let $L$ be any finite lattice and let $a$ be an arbitrary element in
  $L$. We have $\formula{L}{a} = \NC^{1}$.
  \label{thm:NC1}
\end{theorem}
\begin{proof}
  ($\supseteq$) Any lattice with at least 2 elements contains the
  0--1 lattice as a sublattice. Also since $\NC^{1}$ is closed under
  complementation, the class does not change even if the acceptor is
  0.

  ($\subseteq$) Let $F$ be a poly-size formula family over $L$. Let
  $i$ be such that $L$ can be embedded in $\Pi_{i}$. Let $F'$ be the
  formula family over $\Pi_{i}$ that corresponds to $F$. We will now
  construct a log-depth poly-size formula family $F''$ that computes
  the same language as $F'$. We will use $F'$ to denote a formula in
  the family $F'$. Let $v$ be the tree separator of the tree
  corresponding to $F'$. For each $a_{i} \in \Pi_{i}$, we will
  construct two formulae. The first one, say $F^{v}_{1}$, computes the
  value at the root of $F'$ assuming that value at $v$ is $a_{i}$ and
  the other, say $F^{v}_{2}$ computes the value at the node $v$ and
  applies $\checkgeb{a_{i}}$ (See Proposition~\ref{prop:checkgeb}) on
  that value. Then we compute the sub-formula $F^{v}_{1} \wedge
  F^{v}_{2}$. After that we take the lub over all such sub-formulae
  (one for each $a_{i}$). This construction is applied recursively on
  $F^{v}_{1}$ and $F^{v}_{2}$ to obtain a log-depth poly-size formula
  equivalent to $F$.

  Suppose the correct value of the sub-formula of $F'$ rooted at $v$ is
  $a_{i}$. Then the only sub-formulae $F_{1}^{v} \wedge F_{2}^{v}$
  outputting a non-zero value are the ones corresponding to $a_{j} \leq
  a_{i}$.  The non-zero value output by such a sub-formula is $b_{j}$,
  the value obtained at the root when the value of $v$ is $a_{j}$. But
  we know that $b_{i}$, the actual value of the original formula is
  greater than or equal to the value $b_{j}$ of any sub-formula by
  monotonicity of lub and glb. So the topmost lub will always output
  the correct value $b_{i}$.

  The final formula is log-depth, poly-size since the formulae
  $\checkgeb{a}$ have constant depth. Now we can construct an $\NC^{1}$
  circuit from $F''$ by encoding each element in $\Pi_{i}$ in binary
  and replacing each gate in $F''$ by constant-sized circuits
  computing the lub and glb over $\Pi_{i}$.
\end{proof}

\section{Discussion and Conclusion}
We studied the computational power of comparator circuits over bounded
posets. We provide alternative characterizations of $\P$, $\L$, $\NL$
and $\NP$ in terms of comparator circuits.

A natural open problem that comes out of our approach is about a
possible dichotomy between $\P$ and $\CC$ with respect to lattice
structure. More concretely, can we design comparator circuits over
fixed lattices $M_3$ or $N_5$ (or powers of it) for all languages in
$\P$?  Noting that existence of $M_3$ or $N_5$ as a sublattice is a
necessary and sufficient condition for non-distributivity (by the
$M_3$-$N_5$ theorem~\cite{lattice-text}), if we manage to show that
$\XCC{M_3} = \PCC{N_5}{a} = \P$ for any $a \in N_5$, this will show a
dichotomy between $\P$ and $\CC$.

In the context of $\NL$ vs $\L$, there are two open problems. Firstly,
it will also be interesting to see if a dichotomy theorem holds, with
respect to the lattice structure. Secondly, we note that the upper
bound of $\NL$ for the case of skew comparator circuits over finite
lattices, uses the embeddability into partition lattices. The power of
skew comparator circuits over finite bounded posets is unclear. It is
not even clear whether they compute only languages in $\P$.

Cook et al.~\cite{Cook12} proposed the question whether membership
testing for CFLs is in \CC. Our characterization of \CC\ in terms of
distributive lattices leads to a concrete approach towards proving
this. Namely, designing a lattice to decide membership testing for
CFLs and showing that this lattice is distributive.

\paragraph{Acknowledgments: } We thank the anonymous reviewers for
their constructive comments, which helped us improve the paper. In particular, we thank the reviewer who pointed out an error in the proof of earlier Theorem~\ref{thm:SkewP} (where we had erroneously claimed that there exists an $i$, $\XSkewCC{\Pi_i} = \NL$). The reviewer also had outlined an argument the details of which we have incoroporated in this version as the proof of Theorem~\ref{thm:SkewP}.

\bibliographystyle{plain}
\bibliography{lattice}

\newpage
\appendix

\section{Comparator Circuits over Growing Lattices}
\label{app:P-Growing}
We can generalize the comparator circuit model even further by
allowing it to compute over lattices that grow with the size of the
input. If the size of the lattice is polynomial in the size of the
input and if the lattice can be computed by the uniformity machine,
then the languages computed by these circuits are in the class
\P. However, since we have the freedom to change the lattice according
to the size of the input, we may be able to capture the class
\P\ using structurally simpler lattices. It is conceivable that the
class \P\ could be captured by a family of distributive lattices,
while no finite lattice capturing \P\ can be distributive.

In this section, we present a formal definition of comparator circuits
over growing posets and then present a lattice family that captures the
class \P. Then we will show that, even for this simpler lattice, an
embedding to a family of distributive lattices is not possible
(Similar to Theorem~\ref{thm:impossibility}).

\begin{definition}[Comparator Circuits over Growing Bounded Posets]
  A comparator circuit family over a growing bounded poset family
  $\poset{P} = \{P_{n}\}$ with accepting set $\poset{A} = \{A_{n}\}$
  where $A_{n} \subseteq P_{n}$ is a family of circuits $\family{C} =
  {\{ \circuit{C_n} \}}_{n \geq 0}$ where $\circuit{C_n}$ $=$ $(W, G,
  f)$ where $f : W \mapsto (\poset{P_{n}} \setunion \{ (i, g) : 1 \leq i
  \leq n \textrm{ and } g:\Sigma \mapsto \poset{P_{n}} \})$ is a
  comparator circuit. Here $W = \{ w_1, \ldots , w_m \}$ is a set of
  lines and $G$ is an ordered list of gates $(w_i, w_j)$.

  On input $x \in \Sigma^n$, we define the output of the comparator
  circuit $\circuit{C_n}$ as follows.  Each line is initially assigned
  a value according to $f$ as follows. We denote the value of the line
  $w_i$ by $val(w_i)$.  If $f(w) \in \poset{P_{n}}$, then the value is
  the element $f(w)$.  Otherwise $f(w) = (i, g)$ and the initial value
  is given by $g(x_i)$.  A gate $(w_i, w_j)$ (non-deterministically)
  updates the value of the line $w_i$ into $val(w_i) \wedge val(w_j)$
  and the value of the line $w_j$ into $val(w_i) \vee val(w_j)$. The
  values of lines are updated by each gate in $G$ in order and the
  circuit accepts $x$ iff $val(w) = a \in A_{n}$ at the end of the
  computation for some sequence of non-deterministic choices.

  Let $\Sigma$ be any finite alphabet. A comparator circuit family
  \family{C} over a growing bounded poset family \poset{P_n} with an
  accepting $A_{n} \subseteq \poset{P_{n}}$ decides $\lang{L}
  \subseteq \Sigma^*$ iff $\circuit{C_{|x|}}$ correctly decides
  whether $x \in \lang{L}$ for all $x \in \Sigma^*$.

  The circuit family is called \P-uniform if there exists a TM that
  given $1^{n}$ as input runs in \poly(n) time and outputs $P_{n}$,
  $A_{n}$ and $C_{n}$.
\label{def:comp-growing}
\end{definition}

First, we show a lattice family that captures \P.

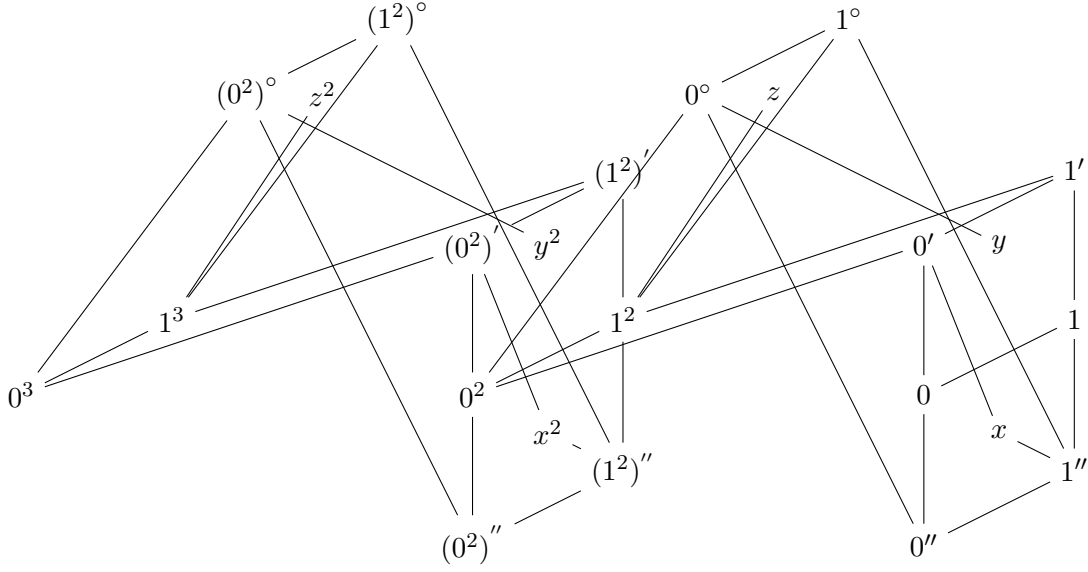
\begin{figure}
\centering
\begin{tikzpicture}
 \node (0)  at (0, 0) {0};
 \node (1)  at (2, 1) {1};
 \node (0pp)at (0, -2) {$0''$};
 \node (1pp) at (2, -1) {$1''$};
 \node (0p) at (0, 2) {$0'$};
 \node (1p) at (2, 3) {$1'$};
 \node (0c) at (-3, 4) {$0^\circ$};
 \node (1c) at (-1, 5) {$1^\circ$};
 \node (x) at (1, -0.5) {$x$};
 \node (y) at (1, 2) {$y$};
 \node (z) at (-2, 4) {$z$};
 
 \node (00)  at (-6, 0) {$0^2$};
 \node (11)  at (-4, 1) {$1^2$};
 \node (00pp)at (-6, -2) {${(0^2)}^{''}$};
 \node (11pp) at (-4, -1) {$(1^2)^{''}$};
 \node (00p) at (-6, 2) {${(0^2)}^{'}$};
 \node (11p) at (-4, 3) {${(1^2)}^{'}$};
 \node (00c) at (-9, 4) {${(0^2)}^\circ$};
 \node (11c) at (-7, 5) {${(1^2)}^\circ$};
 \node (xx) at (-5, -0.5) {$x^2$};
 \node (yy) at (-5, 2) {$y^2$};
 \node (zz) at (-8, 4) {$z^2$};
 
 \node (000) at (-12, 0) {$0^3$};
 \node (111) at (-10, 1) {$1^3$};
 
 \draw (0) -- (1);
 \draw (0p) -- (1p);
 \draw (0pp) -- (1pp);
 \draw (0c) -- (1c);
 \draw (1pp) -- (1);
 \draw (1) -- (1p);
 \draw (0pp) -- (0);
 \draw (0) -- (0p);
 \draw (0pp) -- (0c);
 \draw (1pp) -- (1c);
 \draw (1pp) -- (x);
 \draw (x) -- (0p);
 \draw (y) -- (0c);
 
 \draw (00) -- (0c);
 \draw (00) -- (0p);
 \draw (11) -- (1c);
 \draw (11) -- (1p);
 \draw (11) -- (z);
 
 \draw (00) -- (11);
 \draw (00p) -- (11p);
 \draw (00pp) -- (11pp);
 \draw (00c) -- (11c);
 \draw (11pp) -- (11);
 \draw (11) -- (11p);
 \draw (00pp) -- (00);
 \draw (00) -- (00p);
 \draw (00pp) -- (00c);
 \draw (11pp) -- (11c);
 \draw (11pp) -- (xx);
 \draw (xx) -- (00p);
 \draw (yy) -- (00c);
 
 \draw (000) -- (00c);
 \draw (000) -- (00p);
 \draw (111) -- (11c);
 \draw (111) -- (11p);
 \draw (111) -- (zz); 
 
 \draw (000) -- (111);

\end{tikzpicture}
\caption{A growing poset family for simulating \P}
\label{fig:P-growing}
\end{figure}

\begin{theorem}
  The comparator circuit family over DM completions for the poset
  family in Figure~\ref{fig:P-growing} captures the class \P.
\end{theorem}
\begin{proof}[Proof Sketch]
  We construct a comparator circuit over the poset family in
  Figure~\ref{fig:P-growing} from a layered circuit with NOT gates
  only at the input level. The elements $0^{i}$ and $1^{i}$ in the
  poset correspond to the logical values 0 and 1 at the $i^{th}$ level
  of the circuit. As in the proof of Lemma~\ref{lem:P}, there is a
  sequence of lubs and glbs that creates two copies of the logical
  value at the $i^{th}$ level and then converts them to the
  corresponding value in the $(i+1)^{th}$ level.

  Define $m = |P_{n}|$. The elements of the DM completion of $P_{n}$
  consists of ordered pairs $(A, B)$ where $A, B \subseteq P_{n}$ and
  $A = UP(B)$ and $B = DOWN(A)$. Here $UP(A)$ ($DOWN(A)$) is the set
  of all elements in the poset that are greater (less) than or equal
  to all elements in $A$. Note that in the poset $P_{n}$, if $|A| >
  11$, then we have $DOWN(A) = \phi = B$ and then we have $A =
  P_{n}$. We claim that the DM completion has at most $O(m^{24})$
  elements. Consider an element $(A, B)$ in the DM completion such
  that $|A| > 11$ or $|B| > 11$. If $|A| > 11$, then we have $B =
  \phi$ and therefore $A = P_n$. Similarly, if $|B| > 11$, then we
  have $A = \phi$ and $B = P_n$. Therefore, all elements $(A, B)$
  except $(\phi, P_n)$ and $(P_n, \phi)$ in the DM completion has $|A|
  \leq 11$ and $|B| \leq 11$. This implies that the DM completion has
  at most $O(m^{24})$ elements. To prove the \P-uniformity of the
  comparator circuit family, we have to prove that the DM completion
  can be computed in polynomial time. There exists an algorithm that
  can compute the DM completion of a poset in time polynomial in the
  number of elements in the DM completion \cite{dm-completion}. Since,
  the number of elements in the DM completion of $P_n$ is polynomial
  in $n$, the \P-uniformity of the comparator circuit family follows.
\end{proof}

Now we prove that even this growing lattice family cannot be embedded
into any distributive lattice.

\begin{theorem}
  The poset in Figure~\ref{fig:P-growing} cannot be embedded in any
  distributive lattice.
\end{theorem}
\begin{proof}[Proof Sketch]
  The proof is similar to the proof of
  Theorem~\ref{thm:impossibility}. We use the same labelling
  used in the proof of Theorem~\ref{thm:impossibility}.

  We have $A^{2} = (A'' \cup Y) \cap Z = A^{\circ} \cap Z$ and $B^{2} = (B''
  \cup Y) \cap Z$. Since $B^{2} \supset A^{2}$, we have $x \in
  B^{2}\backslash A^{2}$. So $x \in Z$ and $x \in (B'' \cup Y)\backslash A'$. Now if $x
  \in Y$, then $x \in A'' \cup Y$ and so $x \in A^{2}$. But if $x
  \notin Y$, then $x \in B''$ which implies $x \in A'$ which in turn
  implies $x \in A^{2}$. A contradiction.
\end{proof}

\end{document}